\documentclass[sigconf]{acmart}

\usepackage{booktabs} % For formal tables

\setcopyright{rightsretained}
\acmDOI{10.1145/3289600.3291017}

\acmISBN{978-1-4503-5940-5/19/02} 

\acmConference[WSDM '19] {The Twelfth ACM International Conference on Web Search and Data Mining}{February 11--15, 2019}{Melbourne, VIC, Australia}

\acmYear{2019}
\copyrightyear{2019}
\acmBooktitle{The Twelfth ACM International Conference on Web Search and Data Mining (WSDM '19), February 11--15, 2019, Melbourne, VIC, Australia}

\acmPrice{15.00}

\usepackage{graphicx} % more modern
\usepackage{subfig}

\usepackage{natbib}

\usepackage{algorithm}
\usepackage{algorithmic}
\usepackage{anyfontsize}

\usepackage{xfrac}

\usepackage{amssymb}
\usepackage{amsmath}
\usepackage{bm}
\usepackage{mathtools}
\usepackage{paralist}
\usepackage{bbm}
\usepackage{dsfont}
\usepackage{units}

\graphicspath{ {figures/} }

\usepackage{hyperref}

\usepackage{pgf}
\newtheorem{proposition}{Proposition}

\DeclareMathOperator{\rel}{rel}

\DeclareMathOperator{\click}{c}
\DeclareMathOperator{\Click}{C}

\DeclareMathOperator{\rank}{\mathrm{rk}}
\DeclareMathOperator*{\argmax}{argmax}

\newcommand{\am}[1]{\textcolor{blue}{(AA: #1)}}

\newcommand{\xw}[1]{\textcolor{red}{(XW: #1)}}

\begin{document} 

\title{Estimating Position Bias without Intrusive Interventions}

%%%%%%%True Authors
% unless you've provided the [accepted] option to the icml2013
\author{Aman Agarwal}
\affiliation{%
  \institution{Cornell University}
  \streetaddress{Gates Hall}
  \city{Ithaca}
  \state{NY}
  \postcode{14853}}
\email{aa2398@cornell.edu}

\author{Ivan Zaitsev}
\affiliation{%
  \institution{Cornell University}
  \streetaddress{Gates Hall}
  \city{Ithaca}
  \state{NY}
  \postcode{14853}}
\email{iz44@cornell.edu}

\author{Xuanhui Wang, Cheng Li, Marc Najork}
\affiliation{%
  \institution{Google Inc.}
  \streetaddress{Street}
  \city{Mountain View}
  \state{CA}
  \postcode{99999}}
\email{{xuanhui, chgli, najork}@google.com}

\author{Thorsten Joachims}
\affiliation{%
  \institution{Cornell University}
  \streetaddress{Gates Hall}
  \city{Ithaca}
  \state{NY}
  \postcode{14853}}
\email{tj@cs.cornell.edu}

\renewcommand{\shortauthors}{A. Agarwal et al.}

\renewcommand{\authors}{Aman Agarwal, Ivan Zaitsev, Xuanhui Wang, Cheng Li, Marc Najork and\,\,\, Thorsten Joachims}

\begin{abstract} 
Presentation bias is one of the key challenges when learning from implicit feedback in search engines, as it confounds the relevance signal. While it was recently shown how counterfactual learning-to-rank (LTR) approaches \cite{Joachims/etal/17a} can provably overcome presentation bias when observation propensities are known, it remains to show how to effectively estimate these propensities. In this paper, we propose the first method for producing consistent propensity estimates without manual relevance judgments, disruptive interventions, or restrictive relevance modeling assumptions. First, we show how to harvest a specific type of intervention data from historic feedback logs of multiple different ranking functions, and show that this data is sufficient for consistent propensity estimation in the position-based model. Second, we propose a new extremum estimator that makes effective use of this data. In an empirical evaluation, we find that the new estimator provides superior propensity estimates in two real-world systems -- Arxiv Full-text Search and Google Drive Search. Beyond these two points, we find that the method is robust to a wide range of settings in simulation studies.
\end{abstract}




%\ccsdesc[500]{Information systems~Learning to rank}


\maketitle

\section{Introduction}

In most information retrieval (IR) applications (e.g., personal search, scholarly search, product search), implicit user feedback (e.g. clicks, dwell time, purchases) is routinely logged and constitutes an abundant source of training data for learning-to-rank (LTR). However, implicit feedback suffers from presentation biases, which can make its naive use as training data highly misleading \cite{Joachims/etal/07a}. In particular, the position at which a result is displayed introduces a strong bias, since higher-ranked results are more likely to be discovered by the user than lower-ranked ones.

It was recently shown that counterfactual inference methods provide a provably unbiased and consistent approach to LTR despite biased data \cite{Joachims/etal/17a}. The key prerequisite for counterfactual LTR is knowledge of the propensity of obtaining a particular feedback signal, which enables unbiased Empirical Risk Minimization (ERM) via Inverse Propensity Score (IPS) weighting. This makes getting accurate propensity estimates a crucial prerequisite for effective unbiased LTR.

In this paper, we propose the first approach for producing consistent propensity estimates without manual relevance judgments, disruptive interventions, or restrictive relevance-modeling assumptions. We focus on propensity estimation under the Position-Based Propensity Model (PBM), 
%but the statistical efficiency and unobtrusiveness of our approach makes it feasible to learn multiple conditional PBM (CPBM) models (e.g. one for navigational queries, one for informational queries).  
where all existing propensity estimation methods have substantial drawbacks. In particular, the conventional estimator for the PBM takes a generative approach \cite{Chuklin/etal/15a} and requires individual queries to repeat many times. This is unrealistic for most ranking applications. To avoid this requirement of repeating queries, Wang et al. \cite{Wang/etal/18}  included a relevance model that is jointly estimated via an Expectation-Maximization (EM) procedure~\cite{Dempster+al:1977}. Unfortunately, defining an accurate relevance model is just as difficult as the learning-to-rank problem itself, and a misspecified relevance model can lead to biased propensity estimates. An alternative to generative click modeling on observational data are estimators that rely on specific randomized interventions -- for example, randomly swapping the result at rank $1$ to any rank $k$ \cite{Joachims/etal/17a}. While this provides provably consistent propensity estimates for the PBM, it degrades retrieval performance and user experience~\cite{Wang/etal/18} and is thus costly. In addition, we find in the following that swap-interventions are statistically rather inefficient. The approach we propose in this paper overcomes the disadvantages of the existing methods, as it does not require repeated queries, a relevance model, or additional interventions. 

The key idea behind our estimation technique is to exploit data from a natural intervention that is readily available in virtually any operational system -- namely that we have implicit feedback data from more than one ranking function. We call this Intervention Harvesting. Since click behavior depends jointly on examination and relevance, we show how to exploit this intervention to control for any difference in overall relevance of results at different positions under the PBM. This makes our approach fundamentally interventional and it is thus consistent analogous to using explicit swap interventions under mild assumptions. However, by leveraging existing data that is readily available in most systems, it does not require additional online interventions and the resulting decrease in user experience. To make efficient use of the interventional data we harvest, we propose a specific extremum estimator -- called AllPairs -- that combines all available data for improved statistical efficiency without the need for a relevance model that could introduce bias. We find that this estimator works well even if the rankers we harvest interventions from are quite similar in their result placements and overall performance, and that it is able to recover the relative propensities globally even if most of the changes in rank are small.

%Counterfactual reasoning is fundamentally based on careful consideration of the data generation process, and existing offline methods \cite{Wang/etal/18, ai2018} fail to do so by not exploiting the differences in the rankers used to generate the click logs. Consequently, they need to estimate the relevance of individual query-result pairs which requires making unnecessary modeling assumptions about result relevance. In contrast, due to our relevance control, our method has drastically fewer relevance parameters (independent of the number of queries and documents) and achieves consistent estimation without any such relevance assumptions.

%Finally, we conduct synthetic experiments to demonstrate the effectiveness of our method. %Finally, we conduct synthetic and real-world experiments to demonstrate the effectiveness of our method.   

\section{Related Work}

% Paragraph 1: Learning to rank from implicit feedback. Bias, pairwise preference approaches. See WSDM2017 Unbiased LTR paper.

% Paragraph 2: Counterfactual learning and evaluation. Summary of propensity weighting approaches to LTR. Also older work by Langford, Lihong Li, Miro Dudic on unbiased evaluation. 

% Paragraph 3: Propensity estimation. Google EM paper. Swap 1-k experiment. SIGIR 2016 Google paper. SIGIR 2018 paper by Bruce Croft. Maybe mention \cite{Carterette:2018} as another example of interventions, but with (other) goal of reducing zero propensities.

% Paragraph 4: Click models. Closely related, but not quite the same. But could be adapted to propensity estimation. Give overview of PBM, Cascade, Olivier Chapelle's (see WSDM 2017 Unbiased LTR paper)

Implicit user feedback (e.g. clicks and dwell time) has been widely used to improve search quality. In order to fully utilize the implicit feedback signals for LTR algorithms, various types of bias have to be handled~\cite{Joachims/etal/17a}, e.g., position bias~\cite{Joachims/etal/05a}, presentation bias~\cite{Yue2010a}, and trust bias~\cite{OBrien+Keane:2006, Joachims/etal/05a}. To address this challenge, a large amount of research has been devoted to extracting more accurate signals from click data. For example, some heuristic methods have been proposed to address the position bias by utilizing the pairwise preferences between clicked and skipped documents~\cite{Joachims/etal/05a,Joachims/02c,Joachims/etal/07a}. Though these methods have been found to provide more accurate relevance assessments, their data is still biased. For example, click vs. skip preference tend to reverse the presented order when used for learning~\cite{Joachims/02c} due to their sampling bias.

Recently, Joachims et al.~\cite{Joachims/etal/17a} presented a counterfactual inference framework, which provides a provably unbiased and consistent approach to LTR even with biased feedback data. It requires knowledge of the propensity of obtaining a particular feedback signal, based on which an Inverse Propensity Scoring (IPS) approach can be applied. IPS was developed in causal inference~\cite{Rosenbaum1983} and is a widely accepted technique for handling sampling bias. It has been employed in unbiased evaluation and learning~\cite{Agarwal/etal/17a,Miroslav+al:2011,li2014arxiv,li2011unbiased,Schnabel/etal/16b,Swaminathan/Joachims/15c}. The common assumption in most of these studies is that the propensities are under the system's control and are thus known. In the unbiased LTR setting, however, propensities arise from user behavior and thus need to be estimated.
% This requirement makes it critical to obtain accurate propensity estimates, which is the focus of this paper. On a related note,

The problem of propensity estimation for LTR was already addressed in other works. Wang et al.~\cite{Wang/etal/16} and Joachims et al.~\cite{Joachims/etal/17a} proposed to estimate propensity via randomization experiments. Carterette and Chandar~\cite{Carterette:2018} extend the counterfactual inference framework~\cite{Joachims/etal/17a} by considering the case of evaluating new rankers that can retrieve previously unseen documents. Their methods still rely on interventions, though they make the interventions minimally invasive. To avoid intrusive interventions that degrade the user search experience, Wang et al.~\cite{Wang/etal/18} proposed a regression-based EM algorithm to estimate position bias without interventions. Similarly, Ai et al.~\cite{Ai/etal/18a} presented a framework to jointly learn an unbiased ranker and an unbiased propensity model from biased click data. Both works couple relevance estimation with propensity estimation, which introduces the drawback of potential bias when the relevance model is misspecified. Our work differs from these by controlling for relevance explicitly and without a relevance model, while still avoiding intrusive interventions. 

Another line of research employs click models~\cite{Chuklin/etal/15a} to infer relevance judgments from click logs. Training is typically performed via generative maximum likelihood under specific modeling assumptions about user behavior. There are two classic clicks models: the position-based model (PBM)~\cite{Richardson2007} and the Cascade model~\cite{craswell2008position}. Based on these two models, more advanced models have been developed, including UBM~\cite{dupret2008browsing}, DBN~\cite{Chapelle/Zhang/09}, and CCM~\cite{guo2009ccm}. Our propensity model is based on the PBM. However, we do not use it as a generative model to infer relevance, but instead use interventional techniques to infer propensities even without repeat queries.

% The position bias model assumes that clicks are determined by both document positions and relevance. The Cascade model hypothesizes that users scan documents from the top, and stop when a relevant document is reached. 
%These methods work well on frequent queries in web search, which are issued multiple times. However, their power is limited on long-tail queries or in personal search~\cite{Wang/etal/16}, where it is rare to observe the same query-document pair.

\section{Harvesting Interventions} \label{sec:harvesting}

We approach the propensity estimation task by harvesting implicit interventions from already logged data. In particular, we make use of the fact that we typically have data from multiple historic ranking functions, and we will identify the conditions under which these provide unconfounded intervention data. We focus on the Position-Based Propensity Model (PBM), which we review before defining interventional sets and analyzing their properties.

\subsection{Position-Based Propensity Model} \label{sec:pbm}

The position-based model recognizes that higher-ranked results are more likely to be considered (i.e. discovered and viewed) by the user than results further down the ranking. Suppose that for a particular query $q$, result $d$ is displayed at position $k$. Let $\Click$ be the random variable corresponding to a user clicking on $d$, and let $E$ be the random variable denoting whether the user examines $d$. In our notation, $q$ represents all information about the users, the query context, and which documents the user considers relevant or not. This mean we can denote the relevance of an individual document as a non-random function $\rel(q,d)$ of $q$, where $\rel(q,d)=1$ and $\rel(q,d)=0$ indicates relevant and non-relevant respectively. Then according to the Position-Based Propensity Model (PBM) \cite{Chuklin/etal/15a},
\begin{align*}
    \Pr(\Click=1|q,d,k) & = \Pr(E=1|k) \rel(q,d) \\
                 & = p_k \rel(q,d).
\end{align*}
In this model, the examination probability $p_k := \Pr(E=1|k)$ depends only on the position, and it is identical to the observation propensity \cite{Joachims/etal/17a}. For learning, it is sufficient to estimate relative propensities $\nicefrac{p_k}{p_1}$ for each $k$ \cite{Joachims/etal/17a}, since multiplicative scaling does not change the training objective of counterfactual learning methods (e.g. \cite{Joachims/etal/18a,Joachims/etal/17a,Agarwal/etal/18b,Ai/etal/18a,Swaminathan/Joachims/15c,Swaminathan/Joachims/15d}). Estimating these relative propensities is the goal of this paper. 

Note that one can train multiple PBM models to account for changes in the propensity curve due to context. In this way, the expressiveness of the PBM model can be substantially extended. For example, one can train separate PBM models for navigational vs. informational queries simply by partitioning the data. While training multiple such models is prohibitively expensive when intrusive interventions are required, the intervention harvesting approach we describe below makes training such contextual PBMs feasible since it does not require costly data, is both statistically and computationally efficient, and does not have any parameters that require manual tuning. 

%We also introduce the following Contextual Position-Based Model (CPBM), which is a slight generalization of the PBM model. It allows that the examination probabilities at ranks 2 and beyond additionally depend on a context $x$, i.e.  
%\begin{align*}
%    \forall k \!\ge\! 2: P(\!C\!=\!1|q,d,k,x\!) = P(\!E\!=\!1|k,x\!) P(\!R\!=\!1|q,d\!).
%\end{align*}
%In this model, $x$ can include observable side information about the query $q$ and document $d$, and the following estimation method can be extended to the CPBM. However, we stick to the vanilla PBM for the sake of simplicity.

\subsection{Controlling for Relevance through Swap Interventions} \label{sec:swapexperiment}

The key problem in both propensity estimation and unbiased learning to rank is that we only observe clicks $\Click$, but we never get to observe $E$ (whether a user examined a result) and $\rel(q,d)$ (whether the user found $d$ relevant for $q$) individually. This makes it difficult to attribute the lack of a click to a lack of examination or a lack of relevance. A key idea for overcoming this dilemma without explicit relevance judgments or cumbersome instrumentation (e.g. eye tracking) was proposed in \cite{Joachims/etal/17a}, namely to control for relevance through randomized interventions.

The intervention proposed in \cite{Joachims/etal/17a} is to randomly swap the result in position $1$ with the result in position $k$. We call these $\mathit{Swap}(1,k)$-interventions. Such swaps provide a completely randomized experiment \cite{Imbens/Rubin/15}, meaning that the assignment of the document to a position does not depend on relevance or any covariates (e.g. abstract length, document language). Under these swap interventions, we now get to observe how many clicks position $1$ results get when they stay in position $1$ vs. when they get swapped into position $k$. Since the expected relevance in either condition (i.e. swap vs. not swapped) is the same, any change in clickthrough rate must be proportional to a drop in examination. 

More formally, we model user queries as sampled i.i.d. $q \sim \Pr(Q)$. Whenever a query is sampled, the ranker $f(q)$ sorts the candidate results $d$ for the query and we apply the randomized swap intervention between positions $1$ and some fixed $k$ before the ranking is displayed to the user. As a precursor to the later exposition, suppose that the random swap occurs with a fixed probability $p$ (not necessarily $0.5$), yielding the logged datasets $D^{1,k}_1=(q_1^i,d_1^i,C_1^i)^{n_1}$ (result stayed in position $1$) and $D^{1,k}_k=(q_k^j,d_k^j,C_k^j)^{n_2}$ (result was swapped into position $k$) of sizes $n_1$ and $n_2$ respectively. Here $C_1^i$ denotes whether the document $d_1^i$ placed at position $1$  was clicked or not, and similarly for $C_k^j$. Denote with $\hat{c}_1^{1,k}=\frac{1}{n_1}\sum C_1^i$ the rate of clicks that documents get when they remain in position $1$ and let $\hat{c}_k^{1,k}=\frac{1}{n_2}\sum C_k^j$ be the rate of clicks when they get swapped to position $k$. Then, under the PBM, the relative propensity is equal to the ratio of expected click rates:
\begin{align*} 
    \frac{p_k}{p_1} & = \frac{p_k E_{D^{1,k}_k}\left[\sfrac{\sum \rel(q_k^j,d_k^j)}{n_2}\right]}{p_1 E_{D^{1,k}_1}\left[\sfrac{\sum \rel(q_1^i,d_1^i)}{n_1}\right]} \\
    & = \frac{E_{D^{1,k}_k}\left[\sfrac{\sum p_k \rel(q_k^j,d_k^j)}{n_2}\right]}{E_{D^{1,k}_1}\left[\sfrac{\sum p_1 \rel(q_1^i,d_1^i)}{n_1}\right]} \\
    & = \frac{E_{D^{1,k}_k}\left[\sfrac{\sum \Pr(C_k^j=1|q_k^j,d_k^j,k)}{n_2}\right]}{E_{D^{1,k}_1}\left[\sfrac{\sum\Pr(C_1^i=1|q_1^i,d_1^i,1)}{n_1}\right]} \\
    & = \frac{E_{D^{1,k}_k}\left[\hat{c}_k^{1,k}\right]}{E_{D^{1,k}_1}\left[\hat{c}_1^{1,k}\right]}
\end{align*}
The first equality holds since swaps are completely at random, such that the expected average relevance under each condition is equal and their ratio is one (we need not consider $n_1$ and $n_2$ as random variables). It is thus reasonable to estimate the relative propensity as 
\begin{align*}
    \frac{\hat{p}_k}{\hat{p}_1} = \frac{\hat{c}_k^{1,k}}{\hat{c}_1^{1,k}},
\end{align*}
which is a statistically consistent estimate as the sample size $n$ grows \cite{Joachims/etal/17a}. Analogously, relative propensities can be estimated between any pairs of positions $k$ and $k'$ with $\mathit{Swap}(k,k')$ interventions \cite{Wang/etal/18}.

%\xw{This section is very easy to follow. I think there is a jump between this section and the next one. Some intermediate steps are: (1) The swap can be applied to any position pairs $k$ and $k'$. (2) What is the swap has $n_1$ and no swap has a different $n_2$? In this case we should give a weight $1/n_1$ for swap and $1/n_2$ for non-swap.}

\subsection{Interventional Sets from Multiple Rankers} \label{sec:interventionalsets}

A key shortcoming of the swap experiment is its impact on user experience, since retrieval performance can be degraded quite substantially. This is especially true for swaps between position $1$ and $k$ for large $k$, even though this particular swap experiment makes it easy to estimate propensities relative to position $1$ directly. To overcome this impact on user experience, we now show how to harvest interventions from already existing data under mild assumptions, so that no additional swap experiments are needed. As part of this, we may never see a direct swap between position $1$ and $k$, and we tackle the problem of how to aggregate many local swaps into an overall consistent estimate in Section~\ref{sec:estimator}. 

Our key idea in harvesting swap interventions lies in the use of data from multiple historic rankers. Logs from multiple rankers are typically available in operational systems, since multiple rankers may be fielded at the same time in A/B tests, or when the production ranker is updated frequently. Consider the case where we have data from $m$ historic rankers $F=\{f_1,...,f_m\}$. Furthermore, a crucial condition is that the query distribution must not depend on the choice of ranker $f_i$ (where the query $q$ includes the user's relevance vector $\rel$ in our notation), 
\begin{align}
    \forall f_i: \Pr(Q|f_i) = \Pr(Q) \Rightarrow  \forall q \in Q: \Pr(f_i|q)=\Pr(f_i) \label{eq:independence}
\end{align}
Note that the condition on the left implies that $\Pr(f_i|q)=\Pr(f_i)$ on the right by Bayes rule, which is related to exploration scavenging \cite{Scavenging2008}. Intuitively, this condition avoids that different rankers $f_i$ get different types of queries. For rankers that are compared in an A/B test, this condition is typically fulfilled by construction since the assignment of queries to rankers is randomized. For data from a sequence of production rankers, one needs to be mindful that any temporal covariate shift in the query and user distribution is acceptably small. 

We denote the click log of each ranker $f_i$ with $\mathcal{D}_i =(q_i^j, y_i^j, \click_i^j)^{n_i}$, where $n_i$ is the number of queries that $f_i$ processed.  Here $j\in [n_i]$, $q_i^j$ is a query, $y_i^j=f_i(q_i^j)$ is the presented ranking, and $\click_i^j$ is a vector that indicates for each document click or no click.  Since most retrieval systems only rerank a candidate set of documents in their final stage, we denote $\Omega(q)$ as the candidate set of results for query $q$. We furthermore denote the rank of candidate result $d$ in ranking $y_i^j$ as $\rank(d|y_i^j)$, and we use $\click_i^j(d) \in [0,1]$ to denote whether result $d$ was clicked or not. 

We begin by defining {\em interventional sets} $S_{k,k'}$ of query-document pairs, where one ranking function $f \in F$ puts document $d$ at position $k$ and another ranker $f' \in F$ puts the same document $d$ at position $k'$ for the same query $q$. Let $M$ be some fixed number of top positions for which propensity estimates are desired (e.g. $M=10)$. Then, for each two ranks $k \neq k' \in [M]$, we define the interventional set as 
\begin{align*}
    S_{k,k'} := \{(q,d) : \:& q \in Q, d \in \Omega(q),\;\\ & \exists f,\!f' \,  \rank(d|f\!(q)) \!=\! k \wedge \rank(d|f'\!(q)) \!=\! k'\}
\end{align*}
Intuitively, the pairs in these sets are informative because they receive different treatments or interventions based on the choice of different rankers. But note that we are {\em not} requiring any query to occur multiple times. An interventional set merely reflects that two potential outcomes (i.e. document $d$ either in position $k$ or in position $k'$ for query $q$) were possible. In fact, we only ever observe one factual outcome, while the other outcome remains counterfactual and unobserved. 

The key insight is that for any pair of ranks $(k,k')$, the interventional set contains the query-document pairs $(q,d)$ for which the rank of $d$ was randomly assigned to either $k$ or $k'$ via the choice of ranking function $f_i$. Specifically, there are three possible outcomes depending on the choice of ranking function $f_i \in F$: (a) $f_i$ puts $d$ at rank $k$, (b) $f_i$ puts $d$ at rank $k'$, or (c) $f_i$ puts $d$ at some other rank. In the latter case, the instance is not included in the interventional set $S_{k,k'}$, but the former two cases can be seen as the two conditions of a swap experiment between positions $k$ and $k'$. In this way, we can think about these as virtual swap interventions between ranks $k$ and $k'$ where the randomization comes from the randomized choice of ranking function according to \eqref{eq:independence}.

While this swap experiment is completely randomized under condition $\eqref{eq:independence}$ (i.e. the choice of $f_i$ does not depend on $q$), the assignment is generally not uniform (i.e. $d$ could have a higher probability to be presented in position $k$ than in position $k'$). The weights 
\begin{align*}
    w(q,d,k) := \sum_{i=1}^m n_i\mathds{1}[\rank(d|f_{i}(q)) = k].
\end{align*}
are then used to account for this non-uniformity. Specifically, the weight $w(q,d,k)$ reflects how often document $d$ is ranked at position $k$ given that we have employed each ranking function $f_i$ exactly $n_i$ times. Therefore, the probability of assigning $d$ to $k$ as opposed to $k'$ for query $q$ can be estimated as
\begin{align*}
    P(rank=k|rank=k \mbox{ or } rank=k',q,d) = \frac{w(q,d,k)}{w(q,d,k)+w(q,d,k')}
\end{align*}
We will show in the following how interventional sets can be used to control for unobserved relevance information when estimating propensities.

\subsection{Controlling for Unobserved Relevance through Interventional Sets}

As we had already seen in Section~\ref{sec:swapexperiment}, we need to disentangle two unobserved quantities -- relevance and examination -- when analyzing observed clicks in the PBM. This can be achieved by controlling for relevance through randomization, which was done explicitly in Section~\ref{sec:swapexperiment} via $\mathit{Swap}(1,k)$-interventions. The following shows that interventional sets $S_{k,k'}$ provide analogous control for any pair of ranks $(k,k')$ under condition \eqref{eq:independence}. 

For each interventional set $S_{k,k'}$, with $k \neq k' \in [M]$, we can now define the quantities $\hat{c}_k^{k,k'}$ (and $\hat{c}_{k'}^{k,k'}$) which can be thought of as the rate of clicks in position $k$ (and in position $k'$):
\begin{align*}
    \hat{c}_k^{k,k'} \!\!&:=\!\! \sum_{i=1}^m \!\sum_{j=1}^{n_i} \!\sum_{d \in \Omega(q_i^j)} \!\!\mathds{1}_{[(q_i^j,d) \in S_{k,k'}]}\mathds{1}_{[\rank(d|y_i^j)=k]}\frac{\click_i^j(d)}{w(q_i^j,d,k)}.\\
    \hat{c}_{k'}^{k,k'} \!\!&:=\!\! \sum_{i=1}^m \!\sum_{j=1}^{n_i} \!\sum_{d \in \Omega(q_i^j)} \!\!\mathds{1}_{[(q_i^j,d) \in S_{k,k'}]}\mathds{1}_{[\rank(d|y_i^j)=k']}\frac{\click_i^j(d)}{w(q_i^j,d,k')}.\\
\end{align*}
The definition normalizes the observed clicks by dividing with $w(q_i^j,d,k)$, which accounts for non-uniform assignment probabilities. Note that $w(q_i^j,d,k)$ is non-zero whenever the first indicator is true, such that we never divide a non-zero quantity by zero (and we define $0/0:=0$). Intuitively, $\hat{c}_k^{k,k'}$ and $\hat{c}_{k'}^{k,k'}$ capture the weighted click-through rate at position $k$ and $k'$ restricted to ($k$, $k'$)-interventional (query, document) pairs, where the weights $w(q,d,k)$ account for the imbalance in applying the intervention of putting document $d$ at position $k$ vs $k'$ for query $q$.

%Before we further analyze the properties of interventional sets $S_{k,k'}$ as a completely randomized experiment under condition \eqref{eq:independence} and the use of $\hat{c}_k^{k,k'}$ in the design of propensity estimators, we first illustrate the construction of interventional sets $S_{k,k'}$ through the following examples.

The following shows that $\hat{c}_k^{k,k'}$ and $\hat{c}_{k'}^{k,k'}$ are proportional to the true clickthrough rate at positions $k$ and $k'$ in expectation, conditioned on the number of relevant documents in the interventional set $S_{k,k'}$.

\begin{proposition} \label{prop:expectationc}
For the PBM model, i.i.d. queries $q \sim \Pr(Q)$, and under the condition in \eqref{eq:independence}, the expectations of $\hat{c}_k^{k,k'}$ and $\hat{c}_{k'}^{k,k'}$ are
\begin{align*}
    \mathbb{E}_{q,\click}[\hat{c}_k^{k,k'}] & = p_k  r_{k,k'}\\
    \mathbb{E}_{q,\click}[\hat{c}_{k'}^{k,k'}] & = p_{k'} r_{k,k'}
\end{align*}
where $r_{k,k'}=\mathbb{E}_q\left[\sum_{d \in \Omega(q)} \mathds{1}_{[(q,d) \in S_{k,k'}]} \rel(q,d)\right]$. 
\end{proposition}

\begin{proof}
We only detail the proof for $\hat{c}_k^{k,k'}$, since the proof for $\hat{c}_{k'}^{k,k'}$ is analogous.
\begin{align*}
    &\mathbb{E}_{q,\click}[\hat{c}_k^{k,k'}]  \\
    &=\sum_{i=1}^m \sum_{j=1}^{n_i} \sum_{q \in \mathrm{Q}} \Pr(q) \sum_{d \in \Omega(q)}  \mathds{1}_{[(q,d) \in S_{k,k'}]} \mathds{1}_{[\rank(d|f_i(q))=k]}  \frac{\mathbb{E}_{\click}[\click(d)]}{w(q,d,k)} \\
    &=\sum_{i=1}^m \sum_{j=1}^{n_i} \sum_{q \in \mathrm{Q}} \Pr(q) \sum_{d \in \Omega(q)}  \mathds{1}_{[(q,d) \in S_{k,k'}]} \mathds{1}_{[\rank(d|f_i(q))=k]}  \frac{p_k \rel(q,d)}{w(q,d,k)} \\
    &=p_k \!\sum_{q \in \mathrm{Q}} \!\Pr(q) \!\!\!\sum_{d \in \Omega(q)} \!\!\!\!\mathds{1}_{[(q,d) \in S_{k,k'}]}\rel(q,d) \frac{\sum_{i=1}^m \sum_{j=1}^{n_i}\mathds{1}_{[\rank(d|f_i(q))=k]}}{w(q,d,k)} \\
    &=p_k \mathbb{E}_q [\sum_{d \in \Omega(q)} \mathds{1}_{[(q,d) \in S_{k,k'}]}\rel(q,d) \frac{\sum_{i=1}^m n_i\mathds{1}_{[\mathrm{rk}(d|f_i(q))=k]}}{w(q,d,k)}] \\
    &=p_k \mathbb{E}_q [\sum_{d \in \Omega(q)} \mathds{1}_{[(q,d) \in S_{k,k'}]}\rel(q,d)]
\end{align*}
The first equality follows from the i.i.d. assumption of condition \eqref{eq:independence} and that $y_i^j=f_i(q_i^j)$ by definition, the second from the PBM definition, and the third by taking the inner terms common. 
\end{proof}
The quantity $r_{k,k'}$ is related to the average relevance of the documents in the interventional set $S_{k,k'}$ (also see Section~\ref{sec:allpairs}). While $r_{k,k'}$ is unobserved, it is shared between both $\hat{c}_k^{k,k'}$ and $\hat{c}_{k'}^{k,k'}$. We can thus get the relative propensity between positions $k$ and $k'$ as
\begin{align*}
    \frac{p_k}{p_k'} = \frac{p_k r_{k,k'}}{p_{k'} r_{k,k'}} = \frac{\mathbb{E}_{q,\click}[\hat{c}_k^{k,k'}]}{\mathbb{E}_{q,\click}[\hat{c}_{k'}^{k,k'}]}.
\end{align*}
We are now in a position to define specific estimators for the relative propensities based on the interventional sets.

\section{Propensity Estimators for Interventional Sets} \label{sec:estimator}

This section defines relative propensity estimators that use interventional sets. The first set of estimators, which we call local estimators, are straightforward adaptations of estimators that had been proposed for data from explicit interventions \cite{Joachims/etal/17a,Wang/etal/18}. However, these local estimators ignore much of the available information when harvesting interventions, and we thus develop a new global estimator that exploits information from all interventional sets $S_{k,k'}$.

\subsection{Local Estimators}

Since the click counts $\hat{c}_k^{k,k'}$ and $\hat{c}_{k'}^{k,k'}$ can be treated just like data from an explicit intervention, the same estimators apply. In particular, we observe the interventional sets $S_{1,k}$ and can thus use the same estimator that we previously used for the explicit $\mathit{Swap}(1,k)$ experiment \cite{Joachims/etal/17a}. We call this the {\bf PivotOne} estimator
\begin{align*}
    \frac{\hat{p}_k}{\hat{p}_1} = \frac{\hat{c}_k^{1,k}}{\hat{c}_{1}^{1,k}}
\end{align*}

We can similarly adapt the estimator used in \cite{Wang/etal/18}, which uses a chain of swaps between adjacent positions in the ranking. We call this the {\bf AdjacentChain} estimator
\begin{align*}
    \frac{\hat{p}_k}{\hat{p}_1} = \frac{\hat{c}_2^{1,2}}{\hat{c}_{1}^{1,2}} \cdot \frac{\hat{c}_3^{2,3}}{\hat{c}_{2}^{2,3}} \cdot ... \cdot \frac{\hat{c}_k^{k-1,k}}{\hat{c}_{k-1}^{k-1,k}}
\end{align*}

It is easy to see that both estimators are statistically consistent under mild conditions, most importantly that each relevant interventional set has non-zero support such that $\hat{c}_k^{k,k'}$ and $\hat{c}_{k'}^{k,k'}$ concentrate to their expectations.

\subsection{Global AllPairs Estimator} \label{sec:allpairs}

A key shortcoming of the local estimators is that they only use a small part of the available information, and that they ignore the data from most interventional sets $S_{k,k'}$. To overcome this shortcoming,
we developed a new extremum estimator called {\bf AllPairs} that resembles a maximum likelihood objective over all interventional sets. In order to formulate the training objective of the AllPairs estimator, we first need to define a quantity that is analogous to $\hat{c}_k^{k,k'}$ and $\hat{c}_{k'}^{k,k'}$, but counting the non-click events:
\begin{align*}
    \hat{\neg c}_k^{k,k'} &:= \sum_{i=1}^m \sum_{j=1}^{n_i} \sum_{d \in \Omega(q_i^j)} \mathds{1}_{[(q_i^j,d) \in S_{k,k'}]}\mathds{1}_{[\rank(d|y_i^j)=k]}\frac{1-\click_i^j(d)}{w(q_i^j,d,k)} \\
    \hat{\neg c}_{k'}^{k,k'} &:= \sum_{i=1}^m \sum_{j=1}^{n_i} \sum_{d \in \Omega(q_i^j)} \mathds{1}_{[(q_i^j,d) \in S_{k,k'}]}\mathds{1}_{[\rank(d|y_i^j)=k']}\frac{1-\click_i^j(d)}{w(q_i^j,d,k')}.
\end{align*}
Analogous to $\hat{c}_k^{k,k'}$ and $\hat{c}_{k'}^{k,k'}$, both quantities have the desired expectation.

\begin{proposition} \label{prop:expectationnegc}
For the PBM model, i.i.d. queries $q \sim \Pr(Q)$, and under the condition in \eqref{eq:independence}, the expectations of $\hat{\neg c}_k^{k,k'}$ and $\hat{\neg c}_{k'}^{k,k'}$ are
\begin{align*}
    \mathbb{E}_{q,\click}[\hat{\neg c}_k^{k,k'}] & =N_{k,k'} - p_k r_{k,k'}\\
    \mathbb{E}_{q,\click}[\hat{\neg c}_{k'}^{k,k'}] & = N_{k,k'} - p_k r_{k,k'}
\end{align*}
where $r_{k,k'}=\mathbb{E}_q\left[\sum_{d \in \Omega(q)} \mathds{1}_{[(q,d) \in S_{k,k'}]} \rel(q,d)\right]$ and $N_{k,k'}=\mathbb{E}_q [\sum_{d \in \Omega(q)} \mathds{1}_{[(q,d) \in S_{k,k'}]}]$.
\end{proposition}
\begin{proof}
The proof is analogous to Proposition~\ref{prop:expectationc} and thus omitted.
\end{proof}

The way we constructed the interventional sets, the expected relevances $r_{k,k'}$ are not necessarily normalized to be within the interval $[0,1]$. It is therefore more convenient to instead consider the normalized version of $r_{k,k'}$
\begin{align*}
    \bar{r}_{k,k'} \equiv \frac{r_{k,k'}}{N_{k,k'}} \in [0,1],
\end{align*}
using the definitions from Propositions~\ref{prop:expectationc} and ~\ref{prop:expectationnegc}. Under this normalization, we have that $\mathbb{E} [\hat{c}_k^{k,k'}] =  p_k \bar{r}_{k,k'}N_{k,k'}$ and $\mathbb{E}[\hat{\neg c}_k^{k,k'}] =  (1 - p_k \bar{r}_{k,k'})N_{k,k'}$ and similarly for $\hat{c}_k'^{k,k'}$ and $\hat{\neg c}_k^{k,k'}$. We model this normalized $\bar{r}_{k,k'}$ in the AllPairs estimator. 

We can now formulate the training objective of the AllPairs estimator. The following objective needs to be maximized with respect to $\hat{p}_k \in [0,1]$ and $\hat{r}_{k,k'}=\hat{r}_{k',k} \in [0,1]$ (since $\bar{r}_{k,k'}=\bar{r}_{k',k}$ by definition)
\begin{align*}
(\hat{p},\hat{r}) = \argmax_{p,r}\!\!\!\!\sum_{k\neq k' \in [M]} \!\!\!\!\!\!\hat{c}_k^{k,k'} \log(\hat{p}_k \hat{r}_{k,k'}) + \hat{\neg c}_k^{k,k'} \log(1-\hat{p}_k \hat{r}_{k,k'}).
\end{align*}

This optimization problem can be interpreted as Weighted Cross-Entropy Maximization for estimating the distribution $p_k\bar{r}_{k,k'}$ using weighted samples $\hat{c}_k^{k,k'}$ and $\hat{\neg c}_k^{k,k'}$. The weighting by $N_{k,k'}$ ensures that the contribution of each aggregated click-through sample is proportional to the size of its interventional set. 

From the solution of this optimization problem, $(\hat{p},\hat{r})$, we only need $\hat{p}$ while the matrix $\hat{r}$ can be discarded. To get normalized propensities relative to rank 1, we compute $\nicefrac{\hat{p}_k}{\hat{p}_1}$. Note that the optimization problem is quite small, as it uses only $O(M^2)$ variables and $2*M*(M-1)$ terms in the objective.

Unlike the local estimators, the AllPairs approach integrates all data from every interventional set. We will evaluate empirically how much statistical efficiency is gained by taking this global approach compared to the local estimators. 

\section{Empirical Evaluation}

We take a two-pronged approach to evaluating our intervention-harvesting technique and the estimators. First, we fielded them on two real-world systems -- the Arxiv Full-Text Search Engine and the Google Drive Search -- to gain insight into their practical effectiveness. Second, we augment these real-world experiments with simulation studies using synthetically generated click data, where we can explore the behavior of our method over the whole spectrum of settings (e.g. varying ranker similarity, presentation-bias severity, and click noise). 

\subsection{Real-World Evaluation: Arxiv Full-Text Search}

We conducted a controlled experiment on the Arxiv Full-Text Search Engine\footnote{\url{http://search.arxiv.org:8081/}}, where we compare the results from the AllPairs estimator using Intervention Harvesting with the results from a gold-standard intervention experiment as described below. We used three ranking functions $\{f_1, f_2, f_3\}$ for defining the interventional sets, which were generated by using different learning methods and datasets. 

\paragraph{Gold Standard.} Since we do not know the true propensities for the Arxiv Full-Text Search Engine, we use the $\mathit{Swap}(1,k)$ interventions and estimator described in Section~\ref{sec:swapexperiment} as our gold standard. With probability $0.5$, an incoming query is assigned to generating data towards a swap experiment. For each assigned query, a ranking function $f_i$ is chosen uniformly at random, and its rank 1 is swapped uniformly with rank $k \in \{1,..,21\}$ before it is presented to the user. The $\hat{c}^{1,k}_1$ and $\hat{c}^{1,k}_k$ are computed over all three $f_i$.

\paragraph{Intervention Harvesting.} For the other half of the incoming\,\,\, queries, we uniformly pick a ranking function $f_i$ and present its results without further intervention. From the data in these conditions, we then compute the interventional sets $S_{k,k'}$ for $k \not= k' \in \{1,..,21\}$ and their resulting $\hat{c}_k^{k,k'}$, $\hat{c}_{k'}^{k,k'}$, $\neg \hat{c}_k^{k,k'}$ and $\neg \hat{c}_{k'}^{k,k'}$. These are then used in the AllPairs estimator from Section~\ref{sec:allpairs}.

%\subsubsection{Experiment Set-up.} For each user query that comes in to Arxiv Search, we choose one of three rankers, with equal probability, to determine order of the search results. Then, with a 1/2 probability we perform a swap between the top result and some rank $k$ (chosen uniformly at random from 1 to 21) and display the swapped results to the user. Otherwise, we display results with no swap but making sure to log the rankings of all three rankers. In essence, half of the users are presented with a swap experiment and half are presented with just ranked results, in both cases the ranker is chosen uniformly at random. This lets us compare the swap method with the interventional set method while controlling for as many factors as possible.

\subsubsection{Results.} Data for all six conditions was collected simultaneously between May 14, 2018 and August 1, 2018 to avoid confounding due to shift in the query distribution for maximum validity of the experiment. About 53,000 queries and 25,600 clicks were collected, about half for the Swap Experiment and half for the Interventional Set method. 

The estimated propensity curves for the gold-standard swap experiment and for the AllPairs estimator are shown in Figure~\ref{fig:arxiv}. The shaded region for each curve corresponds to a 95\% confidence interval calculated from 1000 bootstrap samples. As we can see, the curves follow a similar trend for each position and AllPairs mostly lies within the confidence interval of the gold-standard curve. However, the confidence interval for the AllPairs method is substantially tighter than for the swap experiment, indicating that AllPairs is not only less intrusive (no swap interventions) but also statistically more efficient given the same number of queries. 

The following provide further insight into this gain in efficiency. While we conducted interleaving experiments \cite{Chapelle/etal/12a} to confirm that the three ranking functions $\{f_1, f_2, f_3\}$ provide similar ranking accuracy,  we found that the three rankers tend to assign documents to different ranks. In fact, only $7.39\%$ of the documents were ranked at the same rank by all three rankers, such that about $93\%$ of the documents contributed to an interventional set and thus became meaningful training data for AllPairs. Table~\ref{tab:arxiv_matched_total} shows the size of the interventional sets $S_{k,k'}$ for the top 10 positions. While local swaps are most common, the rankers frequently rank the same document at substantially different ranks. 

%We used three of our top performing rankers: PropDCG, PropRankNew, PropRankOld, each of which was trained on several months of Arxiv data. The experiment ran on Arxiv Full-Text Search between May 14, 2018 and August 1, 2018, collecting a total of about 25600 clicks (so $\sim$12800 for the Swap Experiment and $\sim$12800 for the Interventional Set method). \am{Add number of queries} Position matching metrics: pairwise $0.03593$, $0.13734$, $0.07394$, and all three $0.07051$. \am{put in Table?}

\begin{table}[t]
    \caption{Size of the interventional sets $S_{k,k'}$ for Arxiv (showing only top 10 positions).}
    \label{tab:arxiv_matched_total}
    \footnotesize
    \centering
    \begin{tabular}{|c|c|c|c|c|c|c|c|c||c|}
        \hline
         \multicolumn{9}{|c||}{rank $k$} & rank\\
        % \hline
        1 & 2 & 3 & 4 & 5 & 6 & 7 & 8 & 9 & $k'$\\
        \hline
        \hline 13625 & 8340 & 6723 & 5231 & 3966 & 3051 & 2656 & 2274 & 2015 & 2\\
        \hline  -    & 9039 & 7692 & 5994 & 5053 & 3588 & 2675 & 2861 & 2244 & 3\\
        \hline  -    & -    & 8555 & 6994 & 5573 & 4117 & 3368 & 3321 & 2361 & 4\\
        \hline  -    & -    &  -   & 6783 & 5345 & 5040 & 3849 & 3427 & 3614 & 5\\
        \hline  -    & -    &  -   & -    & 6290 & 4809 & 4058 & 4126 & 3489 & 6\\ 
        \hline  -    & -    &  -   & -    & -    & 5466 & 4746 & 3935 & 3294 & 7\\
        \hline  -    & -    &  -   & -    & -    & -    & 5425 & 4092 & 3692 & 8\\
        \hline  -    & -    &  -   & -    & -    & -    & -    & 4258 & 3930 & 9\\
        \hline  -    & -    &  -   & -    & -    & -    & -    & -    & 3719 & 10\\
        \hline
    \end{tabular}
\end{table}

%\subsubsection{Results.} After evaluating propensities using random swap method and propensities using the new method, we get the propensity curves shown in Figure 1. The shaded region for each curve corresponds to a 95\% confidence interval calculated from 1000 bootstrap samples. As we can see, the curves follow a similar trend for each position - even given the relatively small sample size - and lie within each others' 95\% interval. Furthermore, while not shown on the graph, the new method evaluates a curve which converges with significantly fewer samples and has a narrower confidence interval. Given that all factors were controlled for to the best of our ability, the Arxiv experiment suggests that the new method gives us propensities of the same, if not better, quality as the swap experiment.

\begin{figure}
    \centering
    \includegraphics*[width=0.91\linewidth]{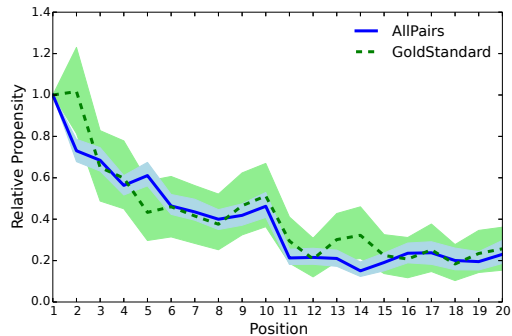}
    \vspace*{-0.3cm}
    \caption{Estimated propensity curves for Arxiv. Shaded regions correspond to a 95\% confidence interval.}
    \label{fig:arxiv}
\end{figure}

\subsection{Real-World Evaluation: Google Drive Search}

We conducted a second real-world experiment on the search for Google Drive. The service uses an overlay to show results as users type. The overlay disappears when a click on the overlay happens. Thus, each query has at most a single click. The overlay displays at most 5 results for each query and all the displayed results are logged with their position information. Again, we compare the results from the AllPairs estimator against a gold-standard intervention experiment.

\paragraph{Gold Standard.} We follow \cite{Wang/etal/18} and use $\mathit{Swap}(k,k+1)$ interventions along with the AdjacentChain estimator as our gold standard. 

\paragraph{Intervention Harvesting.} In an A/B test, users were randomly assigned either to the production ranker $f_1$ or a new ranker $f_2$. Based on these two rankers the interventional sets $S_{k,k'}$ with $k \not= k' \in \{1,..,5\}$ were computed and provided to the AllPairs estimator.

\paragraph{Results.} During two weeks in April 2018, a total of 877,689 queries were collected. Note that data for the gold-standard swap experiment was collected at a different time and for a different ranking function, but is is reasonable to assume that the propensity curve has not changed. The estimated propensity curves for the AllPairs estimator and the gold-standard swap experiment are shown in Figure~\ref{fig:drive}. The AllPairs estimates closely resemble the gold standard as desired. For comparison, we also computed propensity estimates via the regression EM from~\cite{Wang/etal/18} on the same data set as AllPairs. Those estimates are substantially off from the gold standard. Even further off is the naive method of using the empirical clickthrough rate (CTR) as an estimate for $p_k$ without any experimental control for relevance. Overall, we find that Intervention Harvesting with AllPairs provides superior results that are close to the gold standard.

In some ways, the Arxiv experiment and the Google Drive experiment covered two substantially different use cases. For Arxiv, the rankings were substantially different, while for Google Drive the two ranking functions were typically quite close. In fact, the two rankers for Google Drive provided identical top-5 rankings for $75.26\%$ of the queries, which thus did not contribute any useful data to the interventional sets. And even among the queries that lead to different rankings, $74.5\%$ of the results were at the same rank in both ranking functions. This means much less interventional-set data was generated per query, and Table~\ref{tab:drive_matched_total} further shows that most of the swaps were quite local. It is therefore reassuring that AllPairs nevertheless provides accurate estimates, and that it works well for both Arxiv and Google Drive despite these differences.

%\subsubsection{Data Set}
%Our data set is derived from the search logs of Google Drive service. The service uses an overlay to show results as users type. The overlay disappears when a click on the overlay happens. Thus, each query has at most a single click. The overlay displays at most 5 results for each query and all the displayed results are logged with their position information. User clicks are also logged and thus we know which displayed results are clicked in our data set. We discard all queries that do not lead to any click when we process the search logs. 

%We have an A/B experiment framework built for Drive, which is used to do live experiments to compare a new ranker and the production one. The data used in this paper is from a pair of rankers that were run in April 2018 for about 2 weeks. In total, we collected 877,689  queries. Among them, 24.74\% queries have different results from two rankers. In our logs, we record the first position where two rankers gives different results for the same query. Figure~\ref{fig:side_dist} shows the distribution of the first differed positions in our data set. It can be seen that the two rankers differ less at higher positions but more at lower positions. However, there is still a significant amount of queries that differ at position 1. This indicates that there are a good potential to harvest natural interventional pairs from any commercial search engine.

\begin{figure}
    \centering
    \includegraphics*[width=0.91\linewidth]{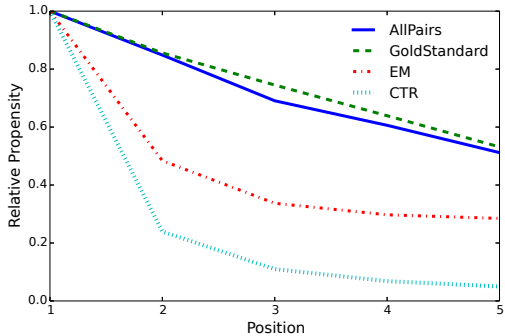}
    \vspace*{-0.3cm}
    \caption{Estimated propensity curves for Google Drive.}
    \label{fig:drive}
\end{figure}

% \subsubsection{Natural Interventional Pairs}

%\begin{figure}
%    \centering
%    \includegraphics*[width=0.91\linewidth,trim={0.5cm 0cm 0.4cm 0.2cm},clip]{side_dist.eps}
%    \vspace*{-0.3cm}
%    \caption{Distribution of first differed position of two rankers on our Drive data set.}
%    \label{fig:side_dist}
%\end{figure}

\begin{table}[t]
    \caption{Size of the interventional sets $S_{k,k'}$ for Google Drive.}
    \label{tab:drive_matched_total}
    \centering
    \begin{tabular}{|c|c|c|c||c|}
        \hline
         \multicolumn{4}{|c||}{rank $k$} & rank\\
        % \hline
        1 & 2 & 3 & 4 & $k'$\\
        \hline
        \hline   19,516 & 2,203 &  579 &  244 & 2\\
        \hline   - & 48,576 & 10,630 & 3,780  & 3\\
        \hline   - & - &  73,170 & 21,037     & 4\\
        \hline   - & - & - & 92,343           & 5\\
        \hline
    \end{tabular}
\end{table}

\subsection{Robustness Analysis: Yahoo LTR Challenge}

While the two real-world experiments provide validation for the applicability of the method, these are just two data-points in a large spectrum of possible settings. We therefore now evaluate the robustness of the method on synthetic click data, where we know the true propensity curve by construction and can control the properties of the data with respect to all relevant parameters (e.g. noise, ranker similarity, data-set size).

\begin{figure}[t]
    \centering
    \includegraphics*[width=0.91\linewidth]{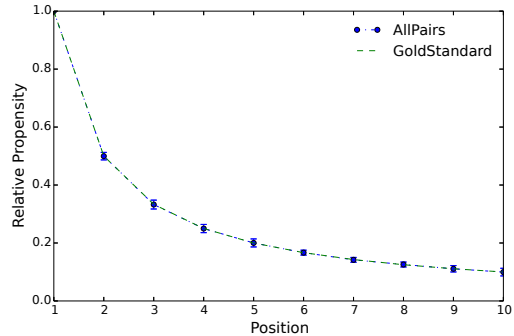}
    \vspace*{-0.3cm}
    \caption{Estimated propensity curve for synthetic click data derived from the Yahoo LTR dataset ($\eta=1$, $\epsilon_-=0.1$, $n_i=99720$, $\mathrm{frac}=0.02$, $\mathrm{overlap}=0.8$).}
    \label{fig:synthetic}
\end{figure}

\paragraph{Experiment Setup.} We generated synthetic click data according to the following methodology that closely matches that in \cite{Joachims/etal/17a}. In particular, we use the Yahoo LTR Challenge data, which comes with manual relevance judgments. Using these relevance judgments, clicks were generated by simulating the Position-Based Model with propensities that decay with the presented rank via $p_r = \big ( \frac{1}{r} \big )^\eta$. The parameter $\eta$ controls the severity of bias, with higher values causing greater position bias. We also introduced noise into the clicks by allowing some irrelevant documents to be clicked. Specifically, an irrelevant document ranked at position $r$ by the production ranker is clicked with probability $p_r$ times $\epsilon_-$ whereas a relevant document is clicked with probability $p_r$. For simplicity (and without loss of generality), we used click logs from two rankers in each experiment setting. Rankers were obtained by training Ranking SVMs on random samples of queries with their manual relevance judgments. The ``similarity'' of the two rankers was controlled by varying the degree of overlap in their respective training sets. We evaluate the estimation accuracy via the Mean Squared Error (MSE) of the estimated inverse relative propensity weights, $MSE= \frac{\sum_{i=1}^{M} (\hat{p}_1/\hat{p}_i - p_1/p_i)^2}{M}$, since these inverse propensity weights better reflect how inaccurate propensity estimates impact the IPS estimator \cite{Joachims/etal/17a}. We estimate propensities up to rank $M=10$. Error bars indicate the standard deviation over $6$ independent runs (except in Figure~\ref{fig:synthetic} as described below). If not mentioned otherwise, the number of simulated queries per ranker is $n_i=99720$ (obtained by $5$ sweeps of the $19944$ queries in the Yahoo LTR training set), and we use $\eta=1$ and $\epsilon_-=0.1$. 

\begin{figure*}[t]
    \centering
    \subfloat[Estimation error with increasing amount of log data for each ranker where every $1$ sweep has 19944 queries. ($\eta=1$, $\epsilon_-=0.1$, $\mathrm{frac}=0.02$, $\mathrm{overlap}=0.8$)]{\includegraphics[width=55mm]{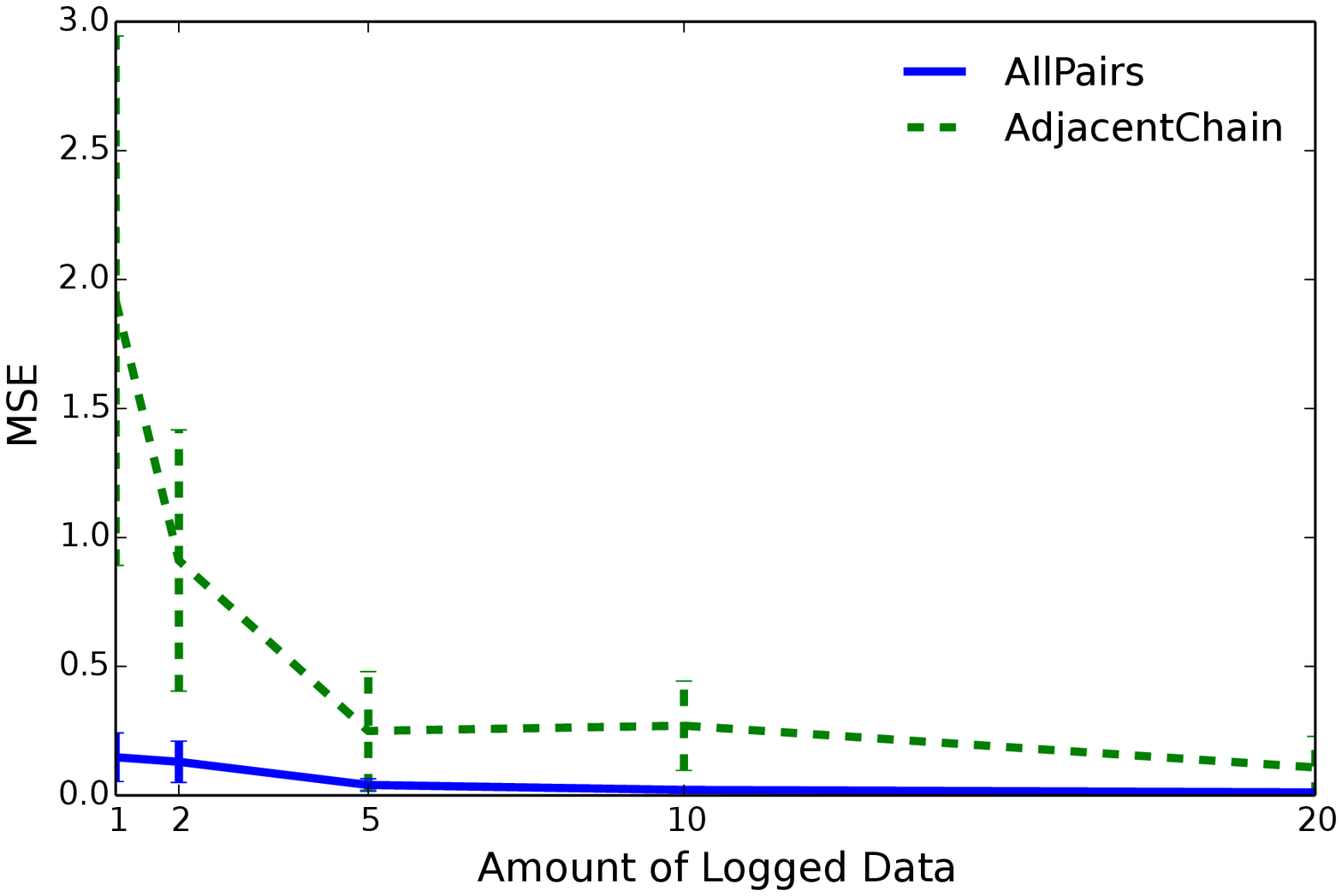}\label{fig:plot_n}}
    \hspace{0.01\textwidth}
    \subfloat[Estimation error with increasing fraction of overlap in the training data for the rankers. ($\eta=1$, $\epsilon_-=0.1$, $n_i=99720$, $\mathrm{frac}=0.02$)]{\includegraphics[width=55mm]{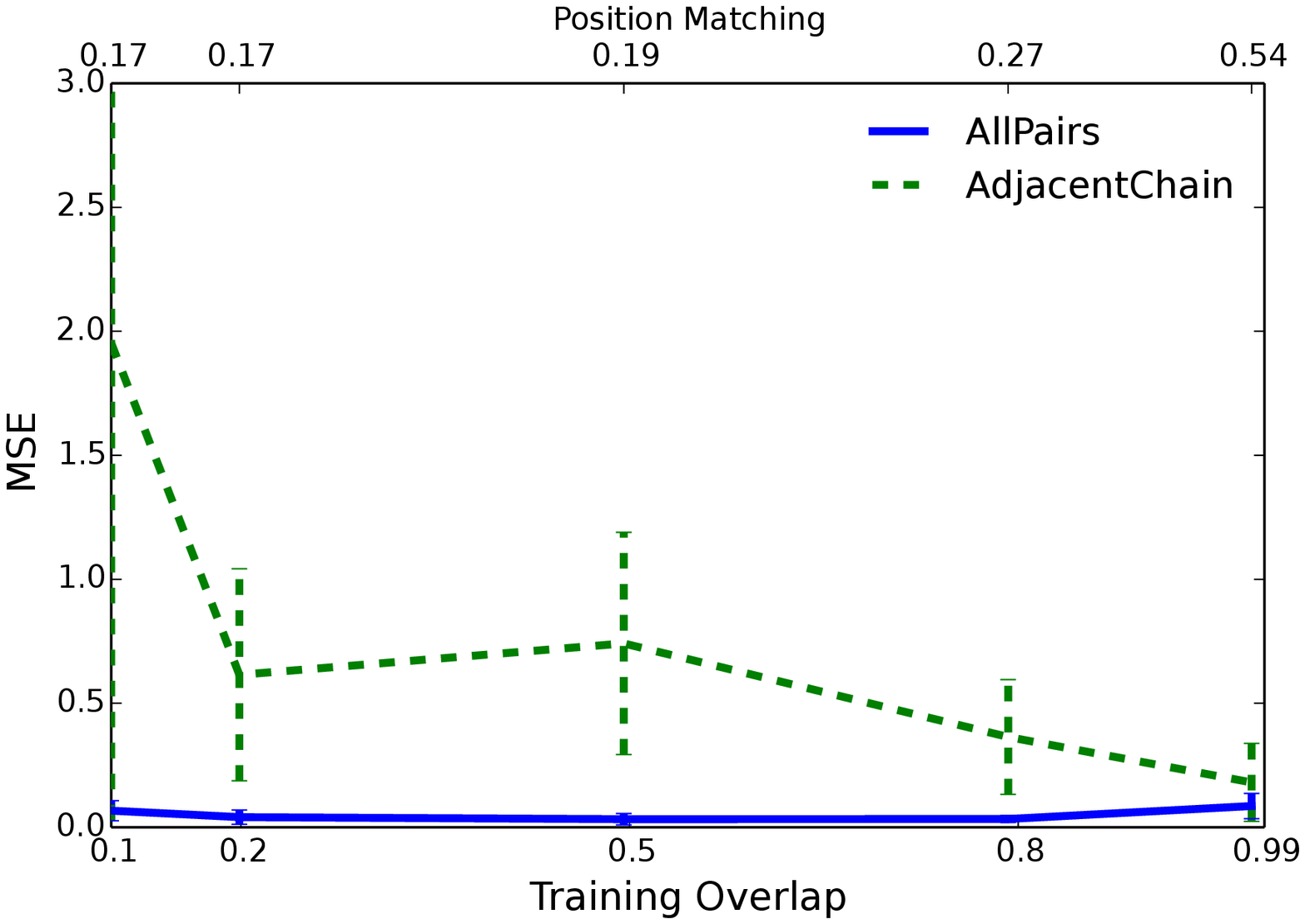}\label{fig:plot_ov}}
    \hspace{0.01\textwidth}
    \subfloat[Estimation error with increasing fraction of the training data used for both rankers. ($\eta=1$, $\epsilon_-=0.1$, $n_i=99720$, $\mathrm{overlap}=0.8$)]{\includegraphics[width=55mm]{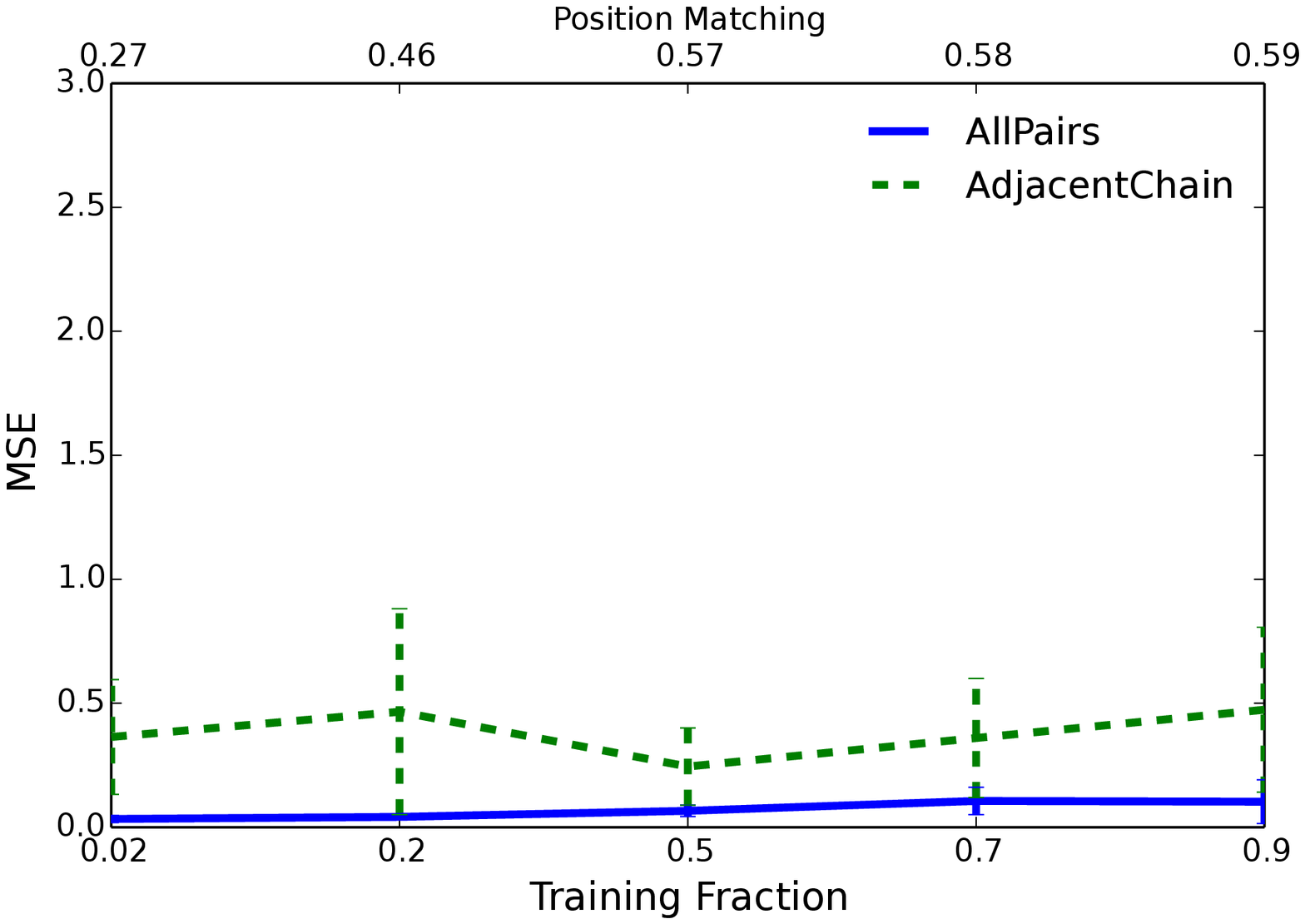}\label{fig:frac}}
    \hspace{0.01\textwidth}
    \subfloat[Estimation error with increasing amount of noise $\epsilon_-$ in the training data. ($\eta=1$, $n_i=99720$, $\mathrm{frac}=0.5$,  $\mathrm{overlap}=0.5$)]{\includegraphics[width=55mm]{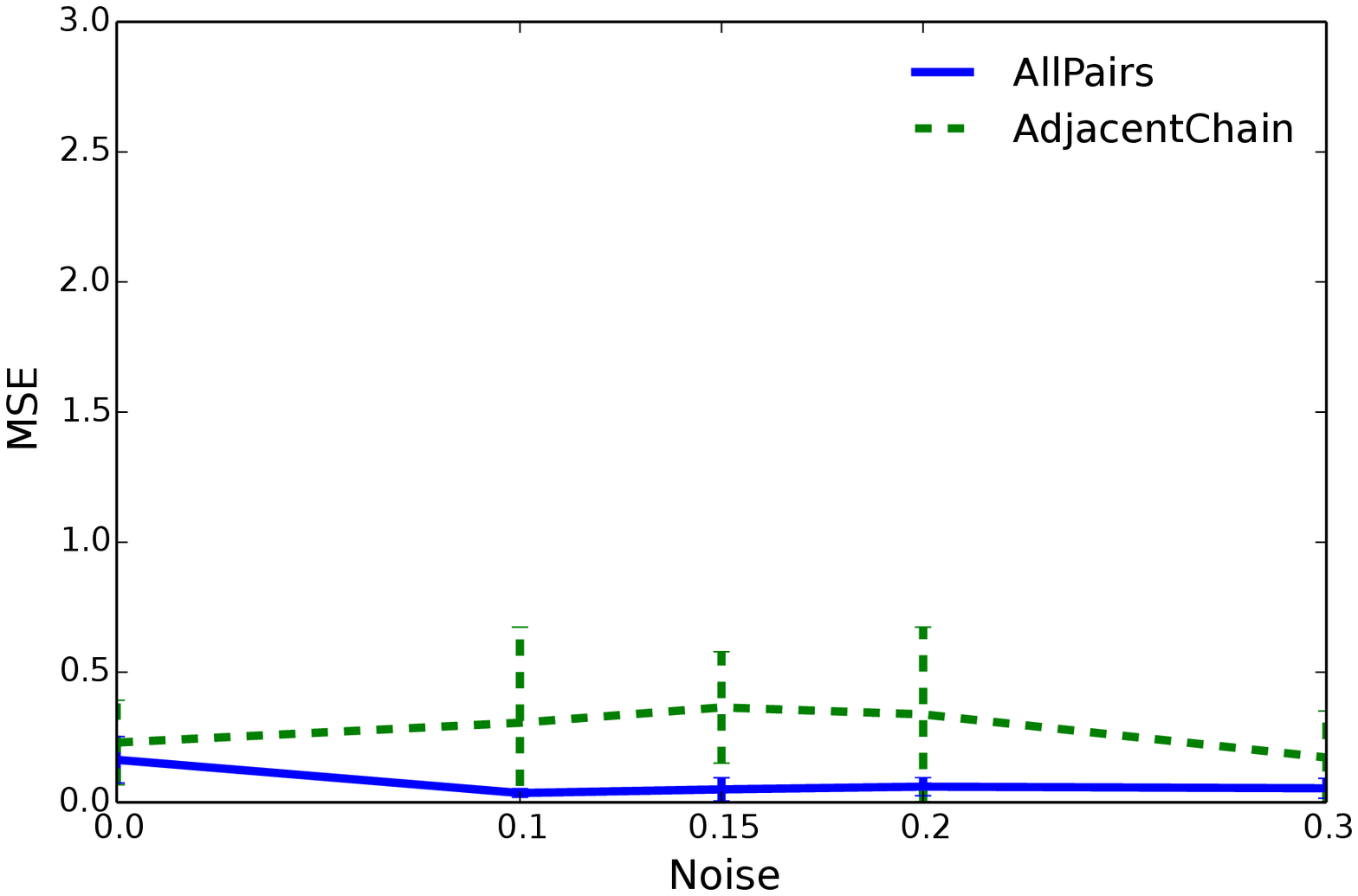}\label{fig:noise}}
    \hspace{0.01\textwidth}
    \subfloat[Estimation error with increasing severity of bias $\eta$  in the training data. ($\epsilon_-=0.1$, $n_i=99720$, $\mathrm{frac}=0.5$,  $\mathrm{overlap}=0.5$)]{\includegraphics[width=55mm]{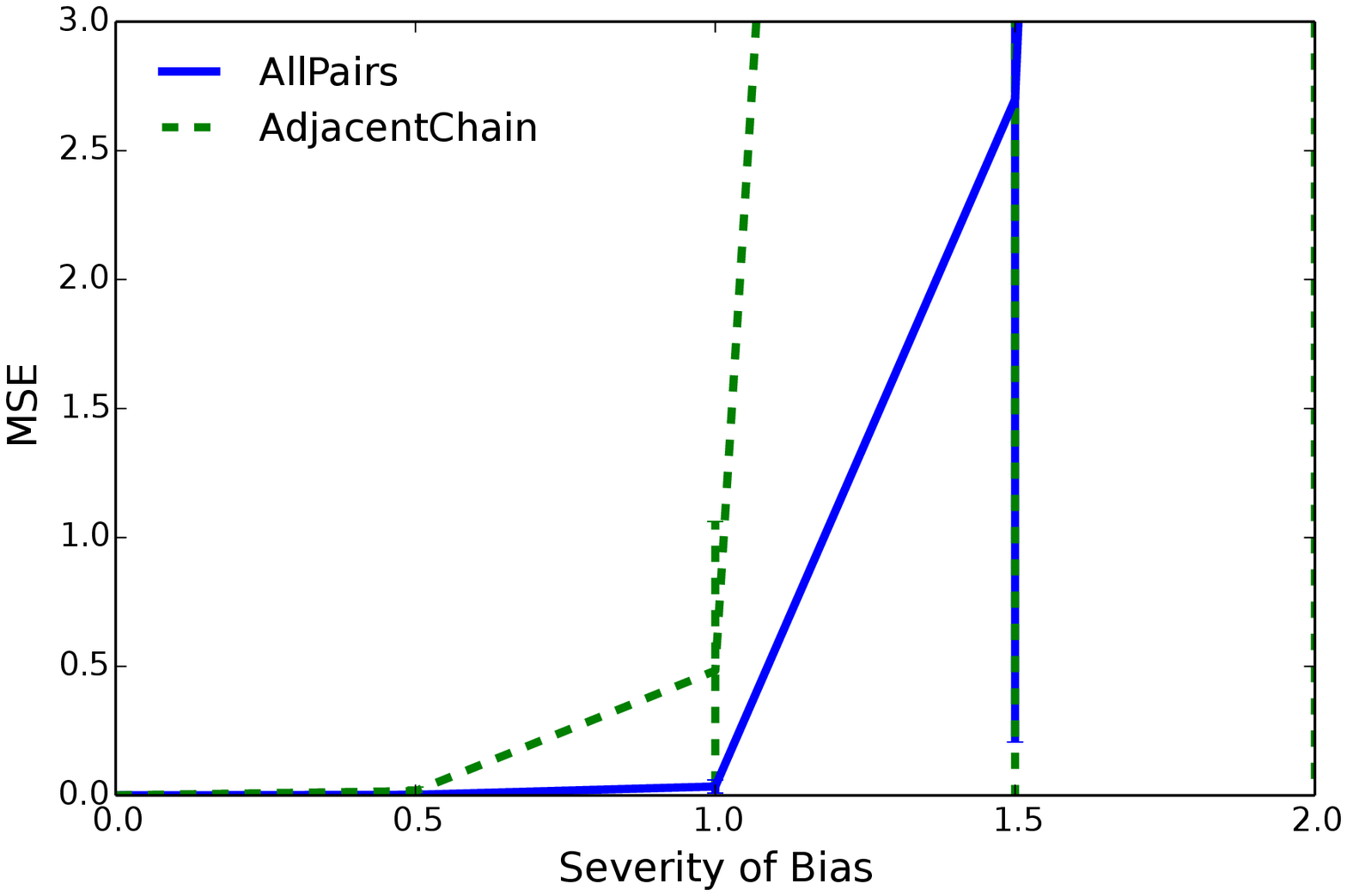}\label{fig:prop}}
    \hspace{0.01\textwidth}
    \subfloat[Estimation error with increasingly unbalanced amounts of log data from $f_1$ vs. $f_2$ where every $1$ sweep has 19944 queries. ($\eta=1$, $\epsilon_-=0.1$, $n_1+n_2=119664$, $\mathrm{frac}=0.02$, $\mathrm{overlap}=0.5$)]{\includegraphics[width=55mm]{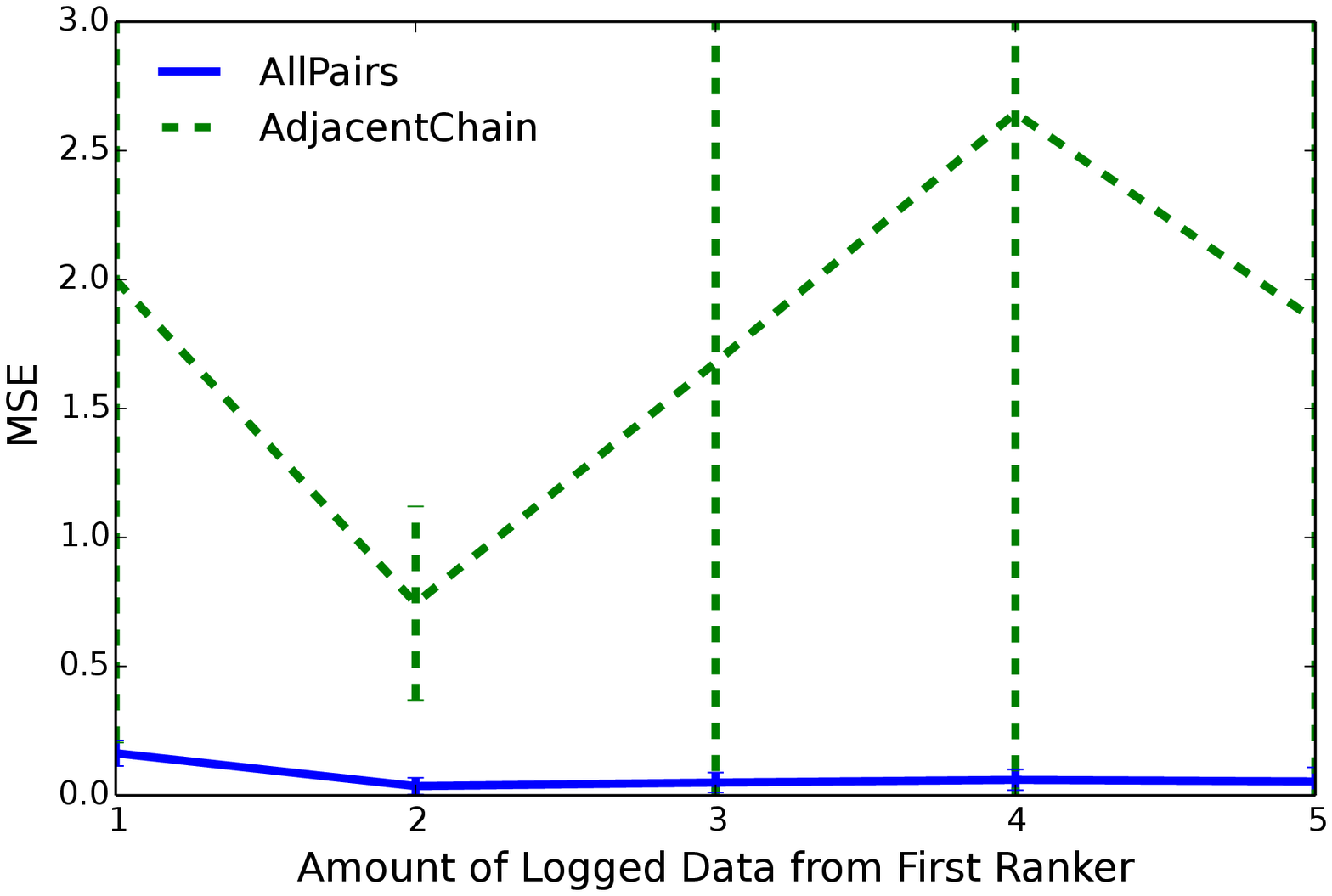}\label{fig:bal}}
    \caption{Robustness experiments on the synthetic click data derived from the Yahoo LTR dataset.}
    \label{fig:synrobust}
\end{figure*}

\subsubsection{Does AllPairs recover the true propensity curve?}

Figure~\ref{fig:synthetic} shows the estimated propensities of AllPairs in comparison to the true propensity curve of ($\eta=1$) that is known by construction. All parameters of the simulation experiment are kept at the defaults as stated above. As expected, AllPairs perfectly recovers the true propensities. The error bars show the $99\%$ confidence intervals that are computed via the standard deviation estimated from re-running the simulation experiment $20$ times.

\subsubsection{How much data is needed?}

While the previous section showed that Intervention Harvesting with the AllPairs estimator converges to the true propensities, we now analyze the speed of convergence. Figure~\ref{fig:synrobust} (a) shows that AllPairs provides good estimates even for modest amounts of data. As a baseline for comparison, we also show the MSE of the AdjacentChain estimator, which is substantially worse and requires at least one order of magnitude more data to achieve the same MSE as AllPairs. We also explored the use of the PivotOne estimator, but do not report its results since they are typically worse than those of AdjacentChain.

% \begin{figure}
%     \centering
%     \includegraphics*[width=0.91\linewidth,trim={0.5cm 0cm 0.4cm 0.2cm},clip]{nsweeps.eps}
%     \vspace*{-0.3cm}
%     \caption{Estimation error with increasing number of sweeps of the dataset during click simulation. ($\eta=1$, $\epsilon_-=0.1$, $\mathrm{overlap}=0.8$)}
%     \label{fig:plot_n}
% \end{figure}

\subsubsection{How different should the ranking functions be?}

If the ranking functions $f_i$ are all identical, then Intervention Harvesting cannot produce any data. So, how different do the rankers $f_i$ need to be? To vary ranker similarity, we trained pairs of rankers with increasing overlap (from $1\%$ to $99\%$) in their training sets. In Figure~\ref{fig:synrobust} (b), we see that the estimation accuracy remains quite robust even as the rankers become increasingly similar due to the overlap in the data they are trained on. The top of the plot shows the similarity in terms of the fraction of documents at the same rank in both rankers averaged across all the queries. As expected, the error goes up when the rankers are very similar since then they tend to put documents at the same position, leading to fewer interventional pairs. Interestingly, the error is also relatively higher when the rankers are too dissimilar. This is because when the candidate sets are larger than $10$, the dissimilarity in the rankers causes many interventions to be discarded since they often go beyond rank $10$. Note that AdjacentChain benefits from ranker similarity, since it focuses the interventional set data on $S_{k,k+1}$. However, AdjacentChain at best matches the performance of AllPairs, but never outperforms it.

% \begin{figure}
%     \centering
%     % \includegraphics*[width=0.91\linewidth,trim={0.5cm 0cm 0.4cm 0.2cm},clip]{plot_ov.png}
%     \includegraphics*[width=0.91\linewidth,trim={0.5cm 0cm 0.4cm 0.2cm},clip]{ov.eps}
%     \vspace*{-0.3cm}
%     \caption{Estimation error with increasing fraction of overlap in the training data for the rankers. ($\eta=1$, $\epsilon_-=0.1$, $\mathrm{sweeps}=5$)}
%     \label{fig:plot_ov}
% \end{figure}

\subsubsection{How important is the quality of the rankers?}

Another way of controlling ranker similarity is to increase the total number of training examples for both. Figure~\ref{fig:synrobust} (c) shows the result of this experiment. Again, AllPairs shows robust performance over the whole spectrum of settings and substantially improves over AdjacentChain.

% \begin{figure}
%     \centering
%     %\includegraphics*[width=0.91\linewidth,trim={0.5cm 0cm 0.4cm 0.2cm},clip]{frac.png}
%     \includegraphics*[width=0.91\linewidth,trim={0.5cm 0cm 0.4cm 0.2cm},clip]{frac.eps}
%     \vspace*{-0.3cm}
%     \caption{Estimation error with increasing fraction of the training data used for both rankers. ($\eta=1$, $\epsilon_-=0.1$, $\mathrm{sweeps}=5$, $\mathrm{overlap}=0.8$)}
%     \label{fig:frac}
% \end{figure}

\subsubsection{How does click noise impact estimation accuracy?}

% \begin{figure}
%     \centering
%     \includegraphics*[width=0.91\linewidth,trim={0.5cm 0cm 0.4cm 0.2cm},clip]{e.eps}
%     \vspace*{-0.3cm}
%     \caption{Estimation error with increasing amount of noise $\epsilon_-$ in the training data. ($\eta=1$, $\mathrm{frac}=0.5$, $\mathrm{sweeps}=5$, $\mathrm{overlap}=0.5$)}
%     \label{fig:noise}
% \end{figure}

Clicks in any real-world setting will be noisy, and we thus want to explore the robustness of AllPairs with respect to noise. Note that noise (which can be seen simply as an alternative vector of relevances $rel(q,d)$ in our model) should have no influence on the propensity estimates via Information Harvesting, as long as the noise is not confounded by rank (which is guaranteed in the PBM). Figure~\ref{fig:synrobust} (d) verifies that Information Harvesting estimates are indeed stable over different levels of $\epsilon_-$ noise.

\subsubsection{How does bias severity impact estimation accuracy?}

Figure~\ref{fig:synrobust} (e) explores the behavior of the estimators when we vary the steepness of the propensity curve via $\eta$. The AllPairs method performs better than AdjacentChain, but the MSE of both methods increases sharply when the propensity curve gets steep. The explanation for this is twofold. First, the increase is partly an artifact of the MSE error measure on the inverse propensities. Steep curves lead to small propensities, which in turns generates large inverse propensities that provide opportunity for large MSE. Second, for steep propensity curves the bottom ranks receive only few clicks such that there is not enough data for estimating their propensities reliably. 

% \begin{figure}
%     \centering
%     \includegraphics*[width=0.91\linewidth,trim={0.5cm 0cm 0.4cm 0.2cm},clip]{p.eps}
%     \vspace*{-0.3cm}
%     \caption{Estimation error with increasing severity of bias $\eta$  in the training data. ($\epsilon_-=0.1$, $\mathrm{sweeps}=5$, $\mathrm{frac}=0.5$, $\mathrm{overlap}=0.5$)}
%     \label{fig:prop}
% \end{figure}

\subsubsection{How robust is the method to imbalanced datasets?}

Finally, we explore the situation where we may have more log data from one ranker than from the other. Figure~\ref{fig:synrobust} (f) shows MSE when we vary the amount of log data for $f_1$ while keeping the total amount of data constant at $119,664$ queries (i.e. $6$ sweeps). The plot shows that AllPairs is robust to such imbalances. 

% \begin{figure}
%     \centering
%     \includegraphics*[width=0.91\linewidth,trim={0.5cm 0cm 0.4cm 0.2cm},clip]{bal.eps}
%     \vspace*{-0.3cm}
%     \caption{Bal}
%     \label{fig:bal}
% \end{figure}

\section{Conclusion}

We presented the idea of Intervention Harvesting, allowing the use of multiple historic loggers for generating interventional data under mild assumptions. We showed how this idea can be used to control for relevance, providing the first method for propensity estimation in the PBM that does not require intrusive interventions, a relevance model, or repeat queries. In particular, we propose the AllPairs estimator for combining all intervention data, which we find to provide superior propensity estimation accuracy compared to existing local estimators over a wide spectrum of settings. Beyond these contributions, the paper opens an interesting space of research questions for how other intervention data can be harvested and what other estimation problems can be addressed in this way.

\section{Acknowledgments}

This research was supported in part by NSF Awards IIS-1615706 and IIS-1513692, an Amazon Research Award, and a gift from Workday. All content represents the opinion of the authors, which is not necessarily shared or endorsed by their respective employers and/or sponsors.

\bibliography{ref}


\begin{thebibliography}{30}

%%% the Web address.

\ifx \showCODEN    \undefined \def \showCODEN     #1{\unskip}     \fi
\ifx \showDOI      \undefined \def \showDOI       #1{#1}\fi
\ifx \showISBNx    \undefined \def \showISBNx     #1{\unskip}     \fi
\ifx \showISBNxiii \undefined \def \showISBNxiii  #1{\unskip}     \fi
\ifx \showISSN     \undefined \def \showISSN      #1{\unskip}     \fi
\ifx \showLCCN     \undefined \def \showLCCN      #1{\unskip}     \fi
\ifx \shownote     \undefined \def \shownote      #1{#1}          \fi
\ifx \showarticletitle \undefined \def \showarticletitle #1{#1}   \fi
\ifx \showURL      \undefined \def \showURL       {\relax}        \fi
% invisible to TeX
\providecommand\bibfield[2]{#2}
\providecommand\bibinfo[2]{#2}
\providecommand\natexlab[1]{#1}
\providecommand\showeprint[2][]{arXiv:#2}

\bibitem[\protect\citeauthoryear{Agarwal, Basu, Schnabel, and Joachims}{Agarwal
  et~al\mbox{.}}{2017}]%
        {Agarwal/etal/17a}
\bibfield{author}{\bibinfo{person}{Aman Agarwal}, \bibinfo{person}{Soumya
  Basu}, \bibinfo{person}{Tobias Schnabel}, {and} \bibinfo{person}{Thorsten
  Joachims}.} \bibinfo{year}{2017}\natexlab{}.
\newblock \showarticletitle{Effective Evaluation using Logged Bandit Feedback
  from Multiple Loggers}. In \bibinfo{booktitle}{\emph{Proc. of the 23rd ACM
  SIGKDD International Conference on Knowledge Discovery and Data Mining}}
  \emph{(\bibinfo{series}{KDD})}. \bibinfo{pages}{687--696}.
\newblock


\bibitem[\protect\citeauthoryear{Agarwal, Zaitsev, and Joachims}{Agarwal
  et~al\mbox{.}}{2018}]%
        {Agarwal/etal/18b}
\bibfield{author}{\bibinfo{person}{Aman Agarwal}, \bibinfo{person}{Ivan
  Zaitsev}, {and} \bibinfo{person}{Thorsten Joachims}.}
  \bibinfo{year}{2018}\natexlab{}.
\newblock \showarticletitle{Counterfactual Learning-to-Rank for Additive
  Metrics and Deep Models}. In \bibinfo{booktitle}{\emph{ICML Workshop on
  Machine Learning for Causal Inference, Counterfactual Prediction, and
  Autonomous Action}} \emph{(\bibinfo{series}{CausalML})}.
\newblock


\bibitem[\protect\citeauthoryear{Ai, Bi, Luo, Guo, and Croft}{Ai
  et~al\mbox{.}}{2018}]%
        {Ai/etal/18a}
\bibfield{author}{\bibinfo{person}{Qingyao Ai}, \bibinfo{person}{Keping Bi},
  \bibinfo{person}{Cheng Luo}, \bibinfo{person}{Jiafeng Guo}, {and}
  \bibinfo{person}{W~Bruce Croft}.} \bibinfo{year}{2018}\natexlab{}.
\newblock \showarticletitle{Unbiased Learning to Rank with Unbiased Propensity
  Estimation}. In \bibinfo{booktitle}{\emph{Proc. of the 41st International ACM
  SIGIR Conference on Research \& Development in Information Retrieval}}
  \emph{(\bibinfo{series}{SIGIR})}. \bibinfo{pages}{385--394}.
\newblock


\bibitem[\protect\citeauthoryear{Carterette and Chandar}{Carterette and
  Chandar}{2018}]%
        {Carterette:2018}
\bibfield{author}{\bibinfo{person}{Ben Carterette} {and}
  \bibinfo{person}{Praveen Chandar}.} \bibinfo{year}{2018}\natexlab{}.
\newblock \showarticletitle{Offline Comparative Evaluation with Incremental,
  Minimally-Invasive Online Feedback}. In \bibinfo{booktitle}{\emph{Proc. of
  the 41st International ACM SIGIR Conference on Research \& Development in
  Information Retrieval}} \emph{(\bibinfo{series}{SIGIR})}.
  \bibinfo{pages}{705--714}.
\newblock


\bibitem[\protect\citeauthoryear{Chapelle, Joachims, Radlinski, and
  Yue}{Chapelle et~al\mbox{.}}{2012}]%
        {Chapelle/etal/12a}
\bibfield{author}{\bibinfo{person}{Olivier Chapelle}, \bibinfo{person}{Thorsten
  Joachims}, \bibinfo{person}{Filip Radlinski}, {and} \bibinfo{person}{Yisong
  Yue}.} \bibinfo{year}{2012}\natexlab{}.
\newblock \showarticletitle{Large-Scale Validation and Analysis of Interleaved
  Search Evaluation}.
\newblock \bibinfo{journal}{\emph{ACM Transactions on Information Systems
  (TOIS)}} \bibinfo{volume}{30}, \bibinfo{number}{1} (\bibinfo{year}{2012}),
  \bibinfo{pages}{6:1--6:41}.
\newblock


\bibitem[\protect\citeauthoryear{Chapelle and Zhang}{Chapelle and
  Zhang}{2009}]%
        {Chapelle/Zhang/09}
\bibfield{author}{\bibinfo{person}{Olivier Chapelle} {and} \bibinfo{person}{Ya
  Zhang}.} \bibinfo{year}{2009}\natexlab{}.
\newblock \showarticletitle{A dynamic bayesian network click model for web
  search ranking}. In \bibinfo{booktitle}{\emph{Proc. of the 18th International
  Conference on World Wide Web}} \emph{(\bibinfo{series}{WWW})}.
  \bibinfo{pages}{1--10}.
\newblock


\bibitem[\protect\citeauthoryear{Chuklin, Markov, and de~Rijke}{Chuklin
  et~al\mbox{.}}{2015}]%
        {Chuklin/etal/15a}
\bibfield{author}{\bibinfo{person}{Aleksandr Chuklin}, \bibinfo{person}{Ilya
  Markov}, {and} \bibinfo{person}{Maarten de Rijke}.}
  \bibinfo{year}{2015}\natexlab{}.
\newblock \bibinfo{booktitle}{\emph{Click Models for Web Search}}.
\newblock \bibinfo{publisher}{Morgan \& Claypool Publishers}.
\newblock


\bibitem[\protect\citeauthoryear{Craswell, Zoeter, Taylor, and Ramsey}{Craswell
  et~al\mbox{.}}{2008}]%
        {craswell2008position}
\bibfield{author}{\bibinfo{person}{Nick Craswell}, \bibinfo{person}{Onno
  Zoeter}, \bibinfo{person}{Michael Taylor}, {and} \bibinfo{person}{Bill
  Ramsey}.} \bibinfo{year}{2008}\natexlab{}.
\newblock \showarticletitle{An Experimental Comparison of Click Position-bias
  Models}. In \bibinfo{booktitle}{\emph{Proc. of the 1st International
  Conference on Web Search and Web Data Mining}}
  \emph{(\bibinfo{series}{WSDM})}. \bibinfo{pages}{87--94}.
\newblock


\bibitem[\protect\citeauthoryear{Dempster, Laird, and Rubin}{Dempster
  et~al\mbox{.}}{1977}]%
        {Dempster+al:1977}
\bibfield{author}{\bibinfo{person}{Arthur~P. Dempster}, \bibinfo{person}{Nan~M.
  Laird}, {and} \bibinfo{person}{Donald~B. Rubin}.}
  \bibinfo{year}{1977}\natexlab{}.
\newblock \showarticletitle{Maximum Likelihood from Incomplete Data via the EM
  Algorithm}.
\newblock \bibinfo{journal}{\emph{Journal of the Royal Statistical Society,
  Series B (Methodological)}}  \bibinfo{volume}{39} (\bibinfo{year}{1977}),
  \bibinfo{pages}{1--38}.
\newblock
Issue 1.


\bibitem[\protect\citeauthoryear{Dud{\'{\i}}k, Langford, and Li}{Dud{\'{\i}}k
  et~al\mbox{.}}{2011}]%
        {Miroslav+al:2011}
\bibfield{author}{\bibinfo{person}{Miroslav Dud{\'{\i}}k},
  \bibinfo{person}{John Langford}, {and} \bibinfo{person}{Lihong Li}.}
  \bibinfo{year}{2011}\natexlab{}.
\newblock \showarticletitle{Doubly Robust Policy Evaluation and Learning}. In
  \bibinfo{booktitle}{\emph{Proc. of the 28th International Conference on
  Machine Learning}} \emph{(\bibinfo{series}{ICML})}.
  \bibinfo{pages}{1097--1104}.
\newblock


\bibitem[\protect\citeauthoryear{Dupret and Piwowarski}{Dupret and
  Piwowarski}{2008}]%
        {dupret2008browsing}
\bibfield{author}{\bibinfo{person}{Georges~E. Dupret} {and}
  \bibinfo{person}{Benjamin Piwowarski}.} \bibinfo{year}{2008}\natexlab{}.
\newblock \showarticletitle{A User Browsing Model to Predict Search Engine
  Click Data from Past Observations}. In \bibinfo{booktitle}{\emph{Proc. of the
  31st Annual International ACM SIGIR Conference on Research and Development in
  Information Retrieval}} \emph{(\bibinfo{series}{SIGIR})}.
  \bibinfo{pages}{331--338}.
\newblock


\bibitem[\protect\citeauthoryear{Guo, Liu, Kannan, Minka, Taylor, Wang, and
  Faloutsos}{Guo et~al\mbox{.}}{2009}]%
        {guo2009ccm}
\bibfield{author}{\bibinfo{person}{Fan Guo}, \bibinfo{person}{Chao Liu},
  \bibinfo{person}{Anitha Kannan}, \bibinfo{person}{Tom Minka},
  \bibinfo{person}{Michael Taylor}, \bibinfo{person}{Yi-Min Wang}, {and}
  \bibinfo{person}{Christos Faloutsos}.} \bibinfo{year}{2009}\natexlab{}.
\newblock \showarticletitle{Click Chain Model in Web Search}. In
  \bibinfo{booktitle}{\emph{Proc. of the 18th International Conference on World
  Wide Web}} \emph{(\bibinfo{series}{WWW})}. \bibinfo{pages}{11--20}.
\newblock


\bibitem[\protect\citeauthoryear{Imbens and Rubin}{Imbens and Rubin}{2015}]%
        {Imbens/Rubin/15}
\bibfield{author}{\bibinfo{person}{G. Imbens} {and} \bibinfo{person}{D.
  Rubin}.} \bibinfo{year}{2015}\natexlab{}.
\newblock \bibinfo{booktitle}{\emph{Causal Inference for Statistics, Social,
  and Biomedical Sciences}}.
\newblock \bibinfo{publisher}{Cambridge University Press}.
\newblock


\bibitem[\protect\citeauthoryear{Joachims}{Joachims}{2002}]%
        {Joachims/02c}
\bibfield{author}{\bibinfo{person}{Thorsten Joachims}.}
  \bibinfo{year}{2002}\natexlab{}.
\newblock \showarticletitle{Optimizing Search Engines Using Clickthrough Data}.
  In \bibinfo{booktitle}{\emph{Proc. of the 8th ACM SIGKDD International
  Conference on Knowledge Discovery and Data Mining}}
  \emph{(\bibinfo{series}{KDD})}. \bibinfo{pages}{133--142}.
\newblock


\bibitem[\protect\citeauthoryear{Joachims, Granka, Pan, Hembrooke, and
  Gay}{Joachims et~al\mbox{.}}{2005}]%
        {Joachims/etal/05a}
\bibfield{author}{\bibinfo{person}{Thorsten Joachims},
  \bibinfo{person}{Laura~A. Granka}, \bibinfo{person}{Bing Pan},
  \bibinfo{person}{Helene Hembrooke}, {and} \bibinfo{person}{Geri Gay}.}
  \bibinfo{year}{2005}\natexlab{}.
\newblock \showarticletitle{Accurately Interpreting Clickthrough Data as
  Implicit Feedback}. In \bibinfo{booktitle}{\emph{Proc. of the 28th Annual
  International ACM SIGIR Conference on Research and Development in Information
  Retrieval}} \emph{(\bibinfo{series}{SIGIR})}. \bibinfo{pages}{154--161}.
\newblock


\bibitem[\protect\citeauthoryear{Joachims, Granka, Pan, Hembrooke, Radlinski,
  and Gay}{Joachims et~al\mbox{.}}{2007}]%
        {Joachims/etal/07a}
\bibfield{author}{\bibinfo{person}{Thorsten Joachims},
  \bibinfo{person}{Laura~A. Granka}, \bibinfo{person}{Bing Pan},
  \bibinfo{person}{Helene Hembrooke}, \bibinfo{person}{Filip Radlinski}, {and}
  \bibinfo{person}{Geri Gay}.} \bibinfo{year}{2007}\natexlab{}.
\newblock \showarticletitle{Evaluating the Accuracy of Implicit Feedback from
  Clicks and Query Reformulations in Web Search}.
\newblock \bibinfo{journal}{\emph{ACM Transactions on Information Systems
  (TOIS)}} \bibinfo{volume}{25}, \bibinfo{number}{2}, Article
  \bibinfo{articleno}{7} (\bibinfo{date}{April} \bibinfo{year}{2007}).
\newblock


\bibitem[\protect\citeauthoryear{Joachims, Swaminathan, and de~Rijke}{Joachims
  et~al\mbox{.}}{2018}]%
        {Joachims/etal/18a}
\bibfield{author}{\bibinfo{person}{Thorsten Joachims}, \bibinfo{person}{Adith
  Swaminathan}, {and} \bibinfo{person}{Maarten de Rijke}.}
  \bibinfo{year}{2018}\natexlab{}.
\newblock \showarticletitle{Deep Learning with Logged Bandit Feedback}. In
  \bibinfo{booktitle}{\emph{6th International Conference on Learning
  Representations}} \emph{(\bibinfo{series}{ICLR})}.
\newblock


\bibitem[\protect\citeauthoryear{Joachims, Swaminathan, and Schnabel}{Joachims
  et~al\mbox{.}}{2017}]%
        {Joachims/etal/17a}
\bibfield{author}{\bibinfo{person}{Thorsten Joachims}, \bibinfo{person}{Adith
  Swaminathan}, {and} \bibinfo{person}{Tobias Schnabel}.}
  \bibinfo{year}{2017}\natexlab{}.
\newblock \showarticletitle{Unbiased Learning-to-Rank with Biased Feedback}. In
  \bibinfo{booktitle}{\emph{Proc. of the 10th ACM International Conference on
  Web Search and Data Mining,}} \emph{(\bibinfo{series}{WSDM})}.
  \bibinfo{pages}{781--789}.
\newblock


\bibitem[\protect\citeauthoryear{Langford, Strehl, and Wortman}{Langford
  et~al\mbox{.}}{2008}]%
        {Scavenging2008}
\bibfield{author}{\bibinfo{person}{John Langford}, \bibinfo{person}{Alexander
  Strehl}, {and} \bibinfo{person}{Jennifer Wortman}.}
  \bibinfo{year}{2008}\natexlab{}.
\newblock \showarticletitle{Exploration Scavenging}. In
  \bibinfo{booktitle}{\emph{Proc. of the 25th International Conference on
  Machine Learning}} \emph{(\bibinfo{series}{ICML})}.
  \bibinfo{pages}{528--535}.
\newblock


\bibitem[\protect\citeauthoryear{Li, Chen, Kleban, and Gupta}{Li
  et~al\mbox{.}}{2014}]%
        {li2014arxiv}
\bibfield{author}{\bibinfo{person}{Lihong Li}, \bibinfo{person}{Shunbao Chen},
  \bibinfo{person}{Jim Kleban}, {and} \bibinfo{person}{Ankur Gupta}.}
  \bibinfo{year}{2014}\natexlab{}.
\newblock \showarticletitle{Counterfactual Estimation and Optimization of Click
  Metrics for Search Engines}.
\newblock \bibinfo{journal}{\emph{CoRR}}  \bibinfo{volume}{abs/1403.1891}
  (\bibinfo{year}{2014}).
\newblock


\bibitem[\protect\citeauthoryear{Li, Chu, Langford, and Wang}{Li
  et~al\mbox{.}}{2011}]%
        {li2011unbiased}
\bibfield{author}{\bibinfo{person}{Lihong Li}, \bibinfo{person}{Wei Chu},
  \bibinfo{person}{John Langford}, {and} \bibinfo{person}{Xuanhui Wang}.}
  \bibinfo{year}{2011}\natexlab{}.
\newblock \showarticletitle{Unbiased Offline Evaluation of
  Contextual-bandit-based News Article Recommendation Algorithms}. In
  \bibinfo{booktitle}{\emph{Proc. of the 4th International Conference on Web
  Search and Web Data Mining}} \emph{(\bibinfo{series}{WSDM})}.
  \bibinfo{pages}{297--306}.
\newblock


\bibitem[\protect\citeauthoryear{O'Brien and Keane}{O'Brien and Keane}{2006}]%
        {OBrien+Keane:2006}
\bibfield{author}{\bibinfo{person}{Maeve O'Brien} {and} \bibinfo{person}{Mark~T
  Keane}.} \bibinfo{year}{2006}\natexlab{}.
\newblock \showarticletitle{Modeling result--list searching in the World Wide
  Web: The role of relevance topologies and trust bias}. In
  \bibinfo{booktitle}{\emph{Proc. of the 28th Annual Conference of the
  Cognitive Science Society}} \emph{(\bibinfo{series}{CogSci})}.
  \bibinfo{pages}{1--881}.
\newblock


\bibitem[\protect\citeauthoryear{Richardson, Dominowska, and Ragno}{Richardson
  et~al\mbox{.}}{2007}]%
        {Richardson2007}
\bibfield{author}{\bibinfo{person}{Matthew Richardson}, \bibinfo{person}{Ewa
  Dominowska}, {and} \bibinfo{person}{Robert Ragno}.}
  \bibinfo{year}{2007}\natexlab{}.
\newblock \showarticletitle{Predicting Clicks: Estimating the Click-through
  Rate for New Ads}. In \bibinfo{booktitle}{\emph{Proc. of the 16th
  International Conference on World Wide Web}} \emph{(\bibinfo{series}{WWW})}.
  \bibinfo{pages}{521--530}.
\newblock


\bibitem[\protect\citeauthoryear{Rosenbaum and Rubin}{Rosenbaum and
  Rubin}{1983}]%
        {Rosenbaum1983}
\bibfield{author}{\bibinfo{person}{Paul~R. Rosenbaum} {and}
  \bibinfo{person}{Donald~B. Rubin}.} \bibinfo{year}{1983}\natexlab{}.
\newblock \showarticletitle{The central role of the propensity score in
  observational studies for causal effects}.
\newblock \bibinfo{journal}{\emph{Biometrika}} \bibinfo{volume}{70},
  \bibinfo{number}{1} (\bibinfo{year}{1983}), \bibinfo{pages}{41--55}.
\newblock


\bibitem[\protect\citeauthoryear{Schnabel, Swaminathan, Singh, Chandak, and
  Joachims}{Schnabel et~al\mbox{.}}{2016}]%
        {Schnabel/etal/16b}
\bibfield{author}{\bibinfo{person}{Tobias Schnabel}, \bibinfo{person}{Adith
  Swaminathan}, \bibinfo{person}{Ashudeep Singh}, \bibinfo{person}{Navin
  Chandak}, {and} \bibinfo{person}{Thorsten Joachims}.}
  \bibinfo{year}{2016}\natexlab{}.
\newblock \showarticletitle{Recommendations as Treatments: Debiasing Learning
  and Evaluation}. In \bibinfo{booktitle}{\emph{Proc. of the 33rd International
  Conference on Machine Learning}} \emph{(\bibinfo{series}{ICML})}.
  \bibinfo{pages}{1670--1679}.
\newblock


\bibitem[\protect\citeauthoryear{Swaminathan and Joachims}{Swaminathan and
  Joachims}{2015a}]%
        {Swaminathan/Joachims/15c}
\bibfield{author}{\bibinfo{person}{Adith Swaminathan} {and}
  \bibinfo{person}{Thorsten Joachims}.} \bibinfo{year}{2015}\natexlab{a}.
\newblock \showarticletitle{Batch Learning from Logged Bandit Feedback through
  Counterfactual Risk Minimization}.
\newblock \bibinfo{journal}{\emph{Journal of Machine Learning Research (JMLR)}}
   \bibinfo{volume}{16} (\bibinfo{date}{Sep} \bibinfo{year}{2015}),
  \bibinfo{pages}{1731--1755}.
\newblock


\bibitem[\protect\citeauthoryear{Swaminathan and Joachims}{Swaminathan and
  Joachims}{2015b}]%
        {Swaminathan/Joachims/15d}
\bibfield{author}{\bibinfo{person}{Adith Swaminathan} {and}
  \bibinfo{person}{Thorsten Joachims}.} \bibinfo{year}{2015}\natexlab{b}.
\newblock \showarticletitle{The Self-Normalized Estimator for Counterfactual
  Learning}. In \bibinfo{booktitle}{\emph{Proc. of the 28th International
  Conference on Neural Information Processing Systems}}
  \emph{(\bibinfo{series}{NIPS})}. \bibinfo{pages}{3231--3239}.
\newblock


\bibitem[\protect\citeauthoryear{Wang, Bendersky, Metzler, and Najork}{Wang
  et~al\mbox{.}}{2016}]%
        {Wang/etal/16}
\bibfield{author}{\bibinfo{person}{Xuanhui Wang}, \bibinfo{person}{Michael
  Bendersky}, \bibinfo{person}{Donald Metzler}, {and} \bibinfo{person}{Marc
  Najork}.} \bibinfo{year}{2016}\natexlab{}.
\newblock \showarticletitle{Learning to Rank with Selection Bias in Personal
  Search}. In \bibinfo{booktitle}{\emph{Proc. of the 39th International ACM
  SIGIR Conference on Research and Development in Information Retrieval}}
  \emph{(\bibinfo{series}{SIGIR})}. \bibinfo{pages}{115--124}.
\newblock


\bibitem[\protect\citeauthoryear{Wang, Golbandi, Bendersky, Metzler, and
  Najork}{Wang et~al\mbox{.}}{2018}]%
        {Wang/etal/18}
\bibfield{author}{\bibinfo{person}{Xuanhui Wang}, \bibinfo{person}{Nadav
  Golbandi}, \bibinfo{person}{Michael Bendersky}, \bibinfo{person}{Donald
  Metzler}, {and} \bibinfo{person}{Marc Najork}.}
  \bibinfo{year}{2018}\natexlab{}.
\newblock \showarticletitle{Position Bias Estimation for Unbiased Learning to
  Rank in Personal Search}. In \bibinfo{booktitle}{\emph{Proc. of the 11th ACM
  International Conference on Web Search and Data Mining}}
  \emph{(\bibinfo{series}{WSDM})}. \bibinfo{pages}{610--618}.
\newblock


\bibitem[\protect\citeauthoryear{Yue, Patel, and Roehrig}{Yue
  et~al\mbox{.}}{2010}]%
        {Yue2010a}
\bibfield{author}{\bibinfo{person}{Yisong Yue}, \bibinfo{person}{Rajan Patel},
  {and} \bibinfo{person}{Hein Roehrig}.} \bibinfo{year}{2010}\natexlab{}.
\newblock \showarticletitle{Beyond position bias: examining result
  attractiveness as a source of presentation bias in clickthrough data}. In
  \bibinfo{booktitle}{\emph{Proc. of the 19th International Conference on World
  Wide Web}} \emph{(\bibinfo{series}{WWW})}. \bibinfo{pages}{1011--1018}.
\newblock


\end{thebibliography}
\bibliographystyle{ACM-Reference-Format}

\end{document}